\documentclass[lettersize,journal]{IEEEtran}
\usepackage{amsfonts}
\usepackage{textcomp}
\usepackage{xcolor}
\usepackage{fancyhdr}
\usepackage{cite}
\usepackage{graphicx}
\usepackage{psfrag}
\usepackage{subfigure}
\usepackage{url}
\usepackage{stfloats}
\usepackage{amsmath}
\usepackage{array}
\usepackage{fancyhdr}
\usepackage{epsfig}
\usepackage{amssymb}
\usepackage{color}
\usepackage{float}
\usepackage{algorithm}
\usepackage{balance}
\usepackage{algpseudocode} 
\usepackage{caption}
\usepackage[english]{babel}
\usepackage[english]{babel}
\usepackage{amsthm}

\newtheorem{theorem}{Theorem}
\newtheorem{lemma}[theorem]{Lemma}
\DeclareMathOperator*{\argmaxA}{\arg\max} 
\def\BibTeX{{\rm B\kern-.05em{\sc i\kern-.025em b}\kern-.08em
    T\kern-.1667em\lower.7ex\hbox{E}\kern-.125emX}}
\begin{document}

\title{Empirical Risk-aware Machine Learning on Trojan-Horse Detection for Trusted Quantum Key Distribution Networks}
\author{\IEEEauthorblockN{Hong-fu Chou,~\IEEEmembership{Member,~IEEE}, Thang X. Vu,~\IEEEmembership{Senior Member,~IEEE}, Ilora Maity,~\IEEEmembership{Member,~IEEE}, Luis M. Garces-Socarras,~\IEEEmembership{ Member,~IEEE}, Jorge L. Gonzalez-Rios,~\IEEEmembership{Member,~IEEE}, Juan Carlos Merlano-Duncan,~\IEEEmembership{Senior Member,~IEEE}, \dag Sean Longyu Ma,~\IEEEmembership{Member,~IEEE}, Symeon Chatzinotas,~\IEEEmembership{Fellow,~IEEE} and Björn Ottersten,~\IEEEmembership{Fellow,~IEEE} \thanks{This work unrevealed any technical details of IDQ QKD devices under the NDA agreement and is conducted under the project LUQCIA and LUX4QCI funded by the European Union Next Generation EU, in collaboration with the Department of Media, Connectivity, and Digital Policy (SMC), Luxembourg. In addition, this paper is based upon work from COST Action 6G-PHYSEC (CA22168), supported by COST  (European Cooperation in Science and Technology).}\thanks{Hong-fu Chou, Thang X. Vu, Ilora Maity, Luis M. Garces-Socarras, Jorge L. Gonzalez-Rios, Symeon Chatzinotas, and Björn Ottersten are with the Interdisciplinary Centre for Security, Reliability and Trust (SnT), University of Luxembourg, Luxembourg (E-mail of corresponding author: hungpu.chou@uni.lu).}\thanks{Sean Longyu Ma is with \dag School of Computer Science, The University of Auckland, New Zealand}}}
\maketitle
\begin{abstract}
Quantum key distribution (QKD) is a cryptographic technique that leverages principles of quantum mechanics to offer extremely high levels of data security during transmission. It is well acknowledged for its capacity to accomplish provable security. However, the existence of a gap between theoretical concepts and practical implementation has raised concerns about the trustworthiness of QKD networks. In order to mitigate this disparity, we propose the implementation of risk-aware machine learning techniques that present risk analysis for Trojan-horse attacks over the time-variant quantum channel. The trust condition presented in this study aims to evaluate the offline assessment of safety assurance by comparing the risk levels between the recommended safety borderline. This assessment is based on the risk analysis conducted. Furthermore, the proposed trustworthy QKD scenario demonstrates its numerical findings with the assistance of a state-of-the-art point-to-point QKD device, which operates over optical quantum channels spanning distances of 1m, 1km, and 30km. Based on the results from the experimental evaluation of a 30km optical connection, it can be concluded that the QKD device provided prior information to the proposed learner during the non-existence of Eve's attack. According to the optimal classifier, the defensive gate offered by our learner possesses the capability to identify any latent Eve attacks, hence effectively mitigating the risk of potential vulnerabilities. The Eve detection probability is provably bound for our trustworthy QKD scenario.
\end{abstract}

\section{Introduction}
Safety is a fundamental requirement in Cyber-Physical Systems (CPSs), wherein safety can be defined as the absence of circumstances that may lead to fatality, harm, professional ailments, or destruction of equipment or property. It can also be understood as the absence of undesirable hazards that could result in injuries to people or harm to human health, whether that is either directly or indirectly through harm to assets or to the natural world. Numerous research studies have addressed the safety problems and available methodologies pertaining to CPSs. Zio \cite{Zio2016} undertook a comprehensive analysis of the attributes associated with complexity and the various methodologies used in risk assessment pertaining to critical infrastructure needs, without explicitly establishing a direct correlation to CPSs. The actions outlined in \cite{NAIR2014} serve as a fundamental foundation for categorizing the current advancements and emerging discoveries in the field of safety for CPSs. This research does not include the execution of risk control methods, since this aspect is contingent upon the unique characteristics of each application area. Safety assurance, discussed in \cite{Bolbot2019} is commonly described as a comprehensive and methodical process that encompasses all the deliberate and organized steps required to instill sufficient trust that an item, offering, organization, or functional system attains a level of safety that is deemed desirable or bearable. The primary objective of the safety assurance is to provide evidence that the systems are secure for their intended purpose. Additionally, it could potentially be employed to show that sufficient measures have been taken to mitigate risks to a permitted degree. Various safety assurance standards, such as IEC 61508 \cite{bell2006introduction} and ARP 4761 \cite{leveson2014comparison}, use distinct methodologies in the realm of risk management. ARP 4761 places significant emphasis on the discipline of system engineering and adopts a top-down evaluation methodology. In contrast, IEC 61508 focuses primarily on the creation and verification of safety requirements. In contrast, the focus of IEC 62508 revolves around the interface between humans and machines. Furthermore, the Technology Readiness Level (TRL) framework has been developed in accordance with a waterfall system engineering life cycle, providing assistance in the design of novel systems. Any observed disparities may be attributed to divergent preferences and approaches within safety and risk management protocols.

From the discussion in \cite{feng2017trusted}, the social construct of trust primarily focuses on the attributes associated with being trustworthy, specifically in a social context. Trustworthiness is determined by individuals collectively agreeing that the entity being trusted is both morally upright, and will consistently make ethical decisions. Trust among people with regard to a Trusted Platform can be seen as a manifestation of confidence in behavioral trust, as it pertains to the guarantee related to the execution and functioning of such a Trusted Platform.  Platforms use social trust to instil confidence in the processes responsible for collecting and presenting behavioral information. This information, therefore, facilitates the determination of the trustworthiness of a platform. The goal and definition of trusted computing presented in \cite{feng2017trusted} is discussed based on the ISO/IEC 15408 standard: a trusted component, operation, or process is characterized by its ability to exhibit consistent behavior across various operating conditions and its strong resilience against corruption by software, malware, and a specified degree of physical interference. The Trusted Platform employs specialized mechanisms that actively gather evidence of its activity and then provide this proof. This knowledge facilitates the assessment of a platform's trustworthiness. 

The deployment of quantum key distribution (QKD) networks in the future and forthcoming scenarios has garnered significant interest due to their potential to provide ultra-secure communication services. The no-cloning theorem in quantum mechanics guarantees the security of transmitting credential keys via a quantum channel. The analysis of the current state of critical components in quantum networks enables the capacity to link quantum devices across long distances, resulting in significant improvements in communication, network efficiency, and security\cite{Parny_2023}. The authors in \cite{Trinh2018} demonstrate the potential for constructing global quantum networks by transmitting quantum states with essential information using free-space optical (FSO) channels. In order to attain a state of absolute security, QKD is used as a protocol that, in principle, ensures the confidentiality of information sent between two distant nodes by establishing secure keys. This approach has been extensively studied and documented in the literature, as seen in \cite{QKD1} and \cite{QKD2}. QKD has emerged as an extensively investigated quantum communication system \cite{Lajos}. It has been successfully implemented in several communication channels, including both fiber-optic and FSO channels. Long-distance FSO quantum communications have been successfully implemented across very long distances, as shown in \cite{Erven_2008}. Additionally, various experiments have been conducted, such as those presented in \cite{Scheidl2009} and \cite{Fedrizzi_2009}. In the following discourse, we present our methodology, supported by ID Quantique (IDQ), a business based in Switzerland that offers cutting-edge industrial solutions for QKD networks \cite{quantique2001quantis}. This technique entails the establishment of a QKD infrastructure using the BB84 communication protocol \cite{chong2010quantum} or a similar invention. The QKD device is usually implemented using the BB84 protocol, Eve is unable to access the sent key in QKD communication. However, ensuring the precise alignment of practical implementations of QKD systems with their corresponding theoretical requirements is a challenging task. The existence of discrepancies between theory and practice has the potential to create vulnerabilities and undermine the integrity of security measures\cite{jain2016attacks}. Trojan-horse attacks have been recognized as a specific risk, primarily targeting the Bob subsystem. The nature of the attack \cite{Jain2015} is that Eve employs a method of attack against Bob by transmitting luminous Trojan-horse pulses in order to ascertain the specific bases chosen by Alice during the execution of the QKD protocol. The transmission of this information is facilitated by the back-reflected pulses emitted from Alice. In adherence to a general principle, it is imperative for Eve to minimize any interference with the authentic quantum signals transmitted from Alice to Bob, as her primary objective is just to ascertain the foundation settings. In this hypothetical scenario \cite{Jain2015}, when Eve successfully achieves correlations over 48\% using the key obtained through identical forward error code, Alice and Bob are unable to detect the parameters being attacked. Additionally, the QBER of 5\% is significantly lower than the QBER abortion rate of 11\%. It may be concluded from this extreme case that the security of the QKD system has been compromised and does not elicit any significant concern, as the percentage to be deducted throughout the process of privacy amplification is 47.8\%, a value lower than the extent of Eve's understanding. While Eve continues to interfere, the escalating QBER to the abortion rate leads to the complete loss of the key sequence, while Alice and Bob are compelled to engage in re-transmission. In addition to the interference caused by Eve, the presence of a time-invariant quantum channel can also be observed, resulting in a decrease in the key rate reported in the experimental outcomes of the IDQ QKD device when used in long-distance 30km QKD networks. The risks arising from Eve's interference residing in the characteristics of the time-variant quantum channel in QKD networks are taken into consideration. This situation poses a certain level of risk that is contingent upon the QKD key consumption while eavesdropping occurs. As shown in Fig. \ref{fig_img}, the risk can be shown by figuring out a loss function that takes into account what is known about the software-defined network (SDN) controller's assigned traffic, eavesdropping, changes in quantum channel parameters that cause time-varying effects, and what is known about the past from the empirical QBER data. The empirical data are obtained by IDQ DV-QKD Cerberis XGR devices via the optical quantum channel in our single and multiple campus dark fiber experiments. 

The primary objective of Trojan detection \cite{Naveenkumar2021} is to enhance the overall reliability of the system and ensure the integrity of circuits. By emulating the approach taken by Trojan detection networks \cite{Zhang2021}, the deep learning networks encapsulate the unidentified trigger shape and deviations in decision boundaries introduced by backdoors through the acquisition of features derived from adversarial patterns and their characteristics. The application of adversarial perturbations in order to get its imprint is an approach. The introduction of a backdoor alters the decision limits of a network, which are effectively communicated via adversarial perturbations. Inspired by the aforementioned approach, our contributions are summarised as follows:
\begin{enumerate}
    \item In order to establish trustworthy assurance, we initiate risk analysis to provide a trust paradigm that mitigates the gap between theory and practice by taking into account the potential eavesdropping interference on top of the loopholes of practical concern. 
    \item The trust condition not only functions as a borderline for monitoring secure key transmission but also quantifies their perception of trustworthiness, allowing them to be aware of any risks at specific levels, including scenarios involving Trojan-horse attacks where the implementation of the BB84 protocol may include the aforementioned security exceptions.
    \item The proposed approach involves utilizing the category-based Gaussian Mixture Model (GMM) in conjunction with the Kolmogorov-Smirnov (KS) test for assessing goodness-of-fit. This method aims to estimate the posterior probability distribution of the empirical QBER dataset in order to evaluate the potential risks associated with practical QKD systems.    
    \item The defensive mechanism under consideration is based on the principles of the Bayes classifier, with the objective of detecting and mitigating possible attacks by Eve's threat. The Bayes classifier is employed as the optimal classifier in our proposed learner to offer empirical posterior information regarding the probability of eavesdropping occurrences, including Trojan-horse attacks. The Eve detection of the proposed QKD scenario is provably bound by applying Markov inequality. The numerical result is accomplished in an effort to be aware of this risk and deal with the exceptional loophole as discussed in the study presented in \cite{Jain2015}.
    \item The deployment of trustworthy QKD networks enables awareness of risk and validation of the trust condition. This paradigm offers a catalyst for future research in domains such as trusted risk-aware routing and raw key queuing mechanisms inside QKD networks, with the aim of achieving risk prediction through the utilization of online reinforcement learning.
\begin{figure*}[htbp]
\centering
\captionsetup{justification=centering,margin=2cm}
\centerline{\includegraphics[width=0.8\textwidth]{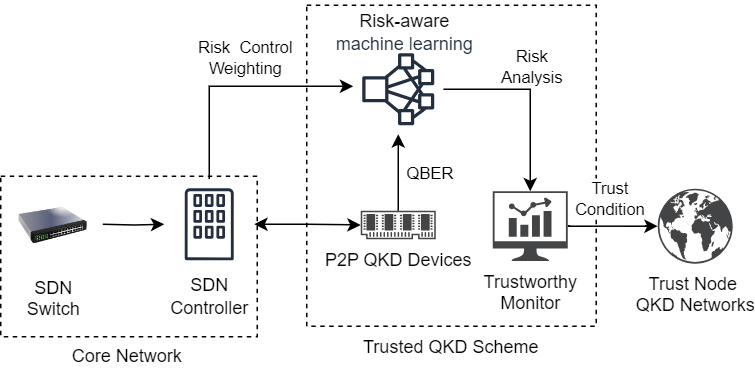}}
\caption{An illustration of the trusted QKD scenario}
\label{fig_img}
\end{figure*}    
\end{enumerate}
The subsequent sections of this paper are structured in the following manner: Section II summarises the related works and provides a comparison of defense mechanisms against Trojan horse attacks.  Section III provides a study of the potential risks associated with the implementation of a viable QKD network. Section IV introduces the proposed learner that can learn the QBER distribution of the QKD device. Section V provides a comprehensive analysis and synthesis of the findings. In Section VI, the numerical results for the trustworthy QKD scenario are presented in order to apply the aforementioned risk analysis and learning procedure. Section VII concludes with a summary of the key insights gained from the study.
\section{Review of related works}
\begin{table*}[h]
\centering
\caption{Summary of Trojan Horse Attacks and Defense Mechanisms in QKD Systems}
\begin{tabular}{|p{4cm}|c|c|c|c|c|c|c|c|c|c|c|}
\hline
\textbf{Theme} & \textbf{\cite{huang2020survey}} & \textbf{\cite{almeida2022ransomware}} & \textbf{\cite{yang2020dynamic}} & \textbf{\cite{sajeed2017invisible}} & \textbf{\cite{yang2015trojan}} & \textbf{\cite{Vinay2018}} & \textbf{\cite{Borisova2020RiskAO}} & \textbf{\cite{Pan2020}} & \textbf{\cite{Navarrete_2022}} & Proposed\\ \hline
Trojan Horse Attack (THA) overview & \checkmark & \checkmark & \checkmark & \checkmark & \checkmark &  &  &  &  &\checkmark\\ \hline
Side-channel attacks (SCA) & \checkmark & \checkmark &  & \checkmark & \checkmark &  &  &  &  &\checkmark\\ \hline
Hardware Trojan incorporation &  & \checkmark & \checkmark &  &  &  &  &  &  &\checkmark\\ \hline
Attenuation and phase modulation issues &  &  &  & \checkmark & \checkmark &  &  &  &  &\\ \hline
Separable coherent state (Gaussian) attacks &  &  &  &  &  & \checkmark &  &  & & \\ \hline
Risk analysis of fiber-optic components &  &  &  &  &  &  & \checkmark &  &  &\checkmark\\ \hline
Spectrum transmission and technical safeguards &  &  &  &  &  &  & \checkmark &  & & \\ \hline
Isolation requirements for secure QKD &  &  &  &  &  &  & \checkmark &  & & \\ \hline
Excess noise analysis in CV-QKD &  &  &  &  &  &  &  & \checkmark &  &\\ \hline
Decoy state vulnerability &  &  &  &  &  &  &  &   &\checkmark &\\ \hline
Finite-key security proofs &  &  &  &  &  &  &  &   &\checkmark &\\ \hline
QBER and vulnerability assessment &  &  &  &  &  &  &  &  & \checkmark &\checkmark\\ \hline
Wavelength and pulse power input effects &  &  &  &  &  &  &  & \checkmark &  &\\ \hline
\end{tabular}
\label{table: checkbox}
\end{table*}
The Discrete-Variable (DV)-QKD protocols, such as BB84, are proven to be resistant to eavesdropping by intercepting the flying qubits and conducting any quantum manipulation on them. The Trojan horse attacks are classified as side-channel attacks (SCA) \cite{huang2020survey}, in which Eve injects her own state into Alice's device and measures the resulting state to infer the key. As a result of the incorporation of hardware trojan horses \cite{almeida2022ransomware} and \cite{ yang2020dynamic}, SCA in QKD systems must be given substantial consideration and cannot be regarded as an inconsequential factor. An SCA, even if it is rudimentary, can significantly compromise the security of a protocol if it is not well-defended. Based on research conducted by \cite{sajeed2017invisible} and \cite{ yang2015trojan}, it has been discovered that the Trojan-horse assault, although experiencing enhanced attenuation and inadequate phase modulation at about 1924nm, has a significant likelihood of eluding detection. The extremely low afterpulsing encountered by Bob's detectors is the cause of this. The majority of the components required to execute this assault are readily available off-the-shelf items. Therefore, in order to avert such assaults, it is critical that operational QKD systems incorporate effective countermeasures. The authors in \cite{Vinay2018} demonstrate that the separable coherent state is the most effective for Eve among a group of multi-mode Gaussian attack states, even when there is thermal noise present. Additionally, the authors establish a limit on the pace at which secret keys may be generated when Eve is allowed to employ any separable state. The analysis evaluates the efficiency of SCA  protection by analyzing the established requirements and measurement results of the researched passive insulating components. The study in \cite{Borisova2020RiskAO} showcases the recorded spectrum transmission spectra of several fiber-optic components.
The graphs provided are applicable for selecting components for quantum key distribution systems, specifically for selecting safeguards against attacks on technical implementation. Additionally, it demonstrates the calculation of the necessary isolation for the QKD system, offering efficient defense against Trojan-horse attacks by leveraging the eavesdropper's utmost technological skills. In order to provide secure QKD, it is necessary for the legitimate parties to have a high level of isolation in the most challenging conditions, while also providing the eavesdropper with the most favorable conditions. This isolation should surpass 150 dB over the whole spectral range that is permitted for transmission by optical fiber. In the study of \cite{Pan2020}, the author proposed a practical method to estimate and correct excess noise in continuous-variable quantum key distribution (CV-QKD) systems. This is achieved by inserting pulses of different wavelengths, taking advantage of the wavelength-dependent property of beam splitters, and considering the existence of non-zero reflection coefficients in the real component. The deviation of shot noise can be accurately estimated using this approach. In conclusion, the authors determined the security bounds of the system by measuring the additional noise generated by Trojan horse attacks. This provides a theoretical basis for the secure transmission of secret keys in the system. For the purpose of quantifying the threshold of extra noise, the authors chose a transmission distance of 50 kilometers as an example. Analysis and simulation experiments indicate that if Trojan horse attacks employ the identical wavelength as the original, the level of crosstalk must be below 0.3.
When employing Trojan horse attacks with a distinct frequency from the original, the magnitude of the attack pulse's power input ought to stay below 0.7 mW. 
Furthermore, the authors in \cite{Navarrete_2022} calculate the limits of security for QKD systems that use decoy states and are vulnerable to Trojan-horse attacks. These new boundaries are far better than earlier estimates. A generic finite-key security proof has been developed for decoy-state-based QKD in the context of probable information leakage from Alice's transmitter. In order to accomplish this, they have used a Cauchy-Schwarz restriction to include the information leakage from the bit/basis and intensity encoding configurations in the security analysis. This restriction necessitates users to limit a singular parameter that encompasses all the flaws, and we have employed innovative concentration constraints to address the consequences of finite-key. Practically, this single parameter may be directly correlated with the level of isolation of Alice's transmitter.
Finally, we compare the defense mechanism employed against the Trojan horse attacks with the QKD system in Table  \ref{table: checkbox}. Our proposed approach gathers QBER data from a QKD device to implement a data-driven technique and precisely assess the vulnerability to Trojan horse attacks. The risk analysis can analytically enhance the assurance of the trusted condition in our trust-based QKD scenario.

\section{Empirical risk analysis for QKD networks}
\begin{figure*}[htbp]
\centering
\centerline{\includegraphics[width=\textwidth]{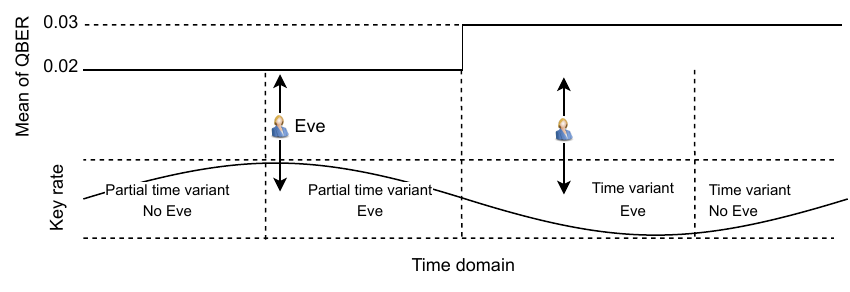}}
\caption{QKD device and Alice subsystem affected by the time-variant quantum channel and Eve's attacks}
\label{fig1_1}
\end{figure*}
The assessment of risk control weighting is primarily determined by the volume of network traffic and the key consumption of QKD networks. It is worth noting that the key consumption rate on demand can be expressed as $\kappa_{ij}=T_{D}/K_{C}$, where $T_{D}$ denotes the number of key demands from the traffic and $K_{C}$ denotes the keys consumed by the QKD node. In \cite{Sharma2021}, the QKD keys are generated by QKD devices and stored in key managers or key pools. Despite a failed point-to-point QKD transmission when the abortion rate is 11\% for our QKD devices, as well as the QBER parameter estimation in \cite{cao2022evolution}, the risk of QKD transmission can be influenced by the ambiguity between prospective Trojan-horse attacks and the presence of time-variant quantum channels, as illustrated in Fig. \ref{fig1_1}. In general, the empirical risk of a QKD network can be calculated by the quantum channel attribute, which is the unknown distribution $\textbf{Q}$ of the QBER probability density function (pdf) and key consumption of the QKD network. We employ the risk loss function $\gamma(\cdot)$ proposed in \cite{Devroye1996} to investigate the risk of QKD networks. 
\begin{itemize}
    \item Let $H(\kappa_{ij})$ denote the pdf of risk control weighting, where $\kappa_{ij}$ represents the key consumption rate on demand during time interval $i$ residing in the codeword of measure QBER $\epsilon_j$. 
    \item The risk $R_{L}$ of the QKD network can be defined as $R_{L}=\int_{\mathcal{W}}^{}H(\kappa_i)\gamma(\theta ,\epsilon  )d L(\epsilon)$ for any quantum channel length $L$ and the space $\mathcal{W}$ with distribution  $\textbf{Q}$. 
    \item The estimated $\widehat{\textbf{Q}}$ is specifically constructed to closely align with $\textbf{Q}$ in order to achieve the objective of minimizing the difference between $\epsilon_j$ and estimated $\widehat{\epsilon_j}$.
\end{itemize}
Due to the complexity of performing this integration, we propose the assessment of the risk using data-driven methodologies, that is empirical computation. The problem formulation to analyze the empirical risk of the QKD network is defined as follows:
\begin{equation}
    \widehat{R_{\varepsilon}}= \frac{1}{N}\sum_{j=1}^{N}\frac{1}{M_{ij}}\sum_{i=1}^{M_{ij}}H(\kappa_{ij})\gamma(\theta ,\epsilon_j  )
\end{equation}
\begin{equation}
    = \frac{1}{N}\sum_{j=1}^{N}H_{M_j}\gamma(\theta ,\epsilon_j  ),
\end{equation}
where
\begin{align}\label{gamma}
\gamma(\theta ,\epsilon_j  ) = 
\delta_{Eve}(\widehat{\textbf{Q}})P\{ \widehat{\textbf{Q}},\textbf{V} \}+ \delta_{var}(\widehat{\textbf{Q}})P\{ \widehat{\textbf{Q}},\textbf{R} \}\notag\\-\delta_{Eve}(\widehat{\textbf{Q}}) \delta_{var}(\widehat{\textbf{Q}})P\{ \widehat{\textbf{Q}},\textbf{V}\},
\end{align}
and the posterior probability of Eavesdropping event and time variance is denoted as $\delta_{Eve}(\widehat{\textbf{Q}})$ and $\delta_{var}(\widehat{\textbf{Q}})$ respectively. The eavesdropping occurrence is denoted as $\textbf{V}$. Furthermore, $\epsilon_j$ denotes the QBER value in the space $\mathcal{W}$ with distribution  $\textbf{Q}$, i.e., $\epsilon\in\textbf{Q}$, and $\forall\theta\in\Omega$ denotes the parameter of the learner in the parameter space $\Omega$. $N$ is the number of experiment samples, and $M_{ij}$ is the number of the key consumption rate changing for the time frame of QBER $\epsilon_j$. The average of $H(\kappa_{ij}),~i=1...M_{ij}$ is denoted as $H_{M_j}$.   

With regard to $\delta_{Eve}(\widehat{\textbf{Q}})$, $\rho$ is the maximum tolerant bit error rate of success decoded before any privacy amplification. The event of Eavesdropping is denoted as $\Delta$. We have
\begin{equation}
\delta_{Eve}(\widehat{\textbf{Q}}) = P\{ \Delta \mid \widehat{\textbf{Q}},\textbf{V} \} \notag   
\end{equation}
\begin{align}
&= P(\epsilon_j=\widehat{\epsilon_j}, \epsilon_j > \rho \mid \widehat{\textbf{Q}}=\lambda(\widehat{\theta}),\textbf{V}) \\
&= P(\epsilon_j > \rho \mid \epsilon_j=\widehat{\epsilon_j},\widehat{\textbf{Q}}=\lambda(\widehat{\theta}),\textbf{V}) P(\epsilon_j=\widehat{\epsilon_j} \mid \widehat{\textbf{Q}}=\lambda(\widehat{\theta}),\textbf{V}) \\
&= P(\widehat{\epsilon_j} > \rho \mid \epsilon_j=\widehat{\epsilon_j},\widehat{\textbf{Q}}=\lambda(\widehat{\theta}),\textbf{V}) P(\epsilon_j=\widehat{\epsilon_j} \mid \widehat{\textbf{Q}}=\lambda(\widehat{\theta}),\textbf{V}),
\end{align}
where $\lambda (\widehat{\theta})$ denotes the learner to estimate the unknown distribution $\textbf{Q}$ of QBER using parameter $\widehat{\theta}$ and $\widehat{\epsilon_j}$ is generated from distribution $\widehat{\textbf{Q}}$, such that  $\widehat{\epsilon_j}\in\widehat{\textbf{Q}}$.  

In addition, quantum communications encompass the process of transmitting quantum states across a quantum channel \cite{Lajos}. The transmission of pre-determined quantum states via a quantum channel, such as an optical fiber or an FSO channel, from a transmitting party represented by Alice to a receiving party represented by Bob involves the propagation of an optical wave in FSO communication. This communication method occurs in an unobstructed environment, which is susceptible to various disturbances leading to time-varying errors. Disturbances such as turbulence, absorption, and scattering contribute to the attenuation of the wave. The aforementioned disruptions exert an influence on the electromagnetic properties, morphology, and orientation of the beam, hence affecting the overall efficacy of the optical connection. The distance of the FSO link is subject to the impact of unpredictable weather phenomena such as haze, rain, and fog. The authors of \cite{ANBARASI2017} examine several methodologies for constructing mathematical models of satellite (classical) channels. Indeed, quantum channels are subject to temporal variations with the inclusion of $T_1$ and $T_2$ fluctuations. Recent experimental investigations \cite{Etxezar2021} have demonstrated that the relaxation time $T_1$ and the dephasing time $T_2$ for superconducting qubits exhibit considerable temporal variations in superconducting quantum computers. This research examines the temporal variation of a quantum channel by utilizing an experimental setup with an optical fiber setup spanning a distance of 30km. The experiment is conducted using cutting-edge IDQ QKD devices \cite{huang2018quantum}. The possible risk of decreasing the raw key rate due to the time-variant impact may include susceptibility to Trojan horse attacks. The event of the quantum time variations is denoted as $\Lambda$. We derive the posterior probability $\delta_{var}(\widehat{\textbf{Q}})$ and $\textbf{R}$ denote the occurrence of time variation as follows:
\begin{align}
\delta_{var}(\widehat{\textbf{Q}})= P\{ 
\Lambda \mid\widehat{\textbf{Q}},\textbf{R}\} 
= P\{\Lambda \mid\widehat{\textbf{Q}}=\lambda(\widehat{\theta}),\textbf{R}\}~~~~~\\
= P(\epsilon_j=\widehat{\epsilon_j}, \epsilon_j > \rho \mid\widehat{\textbf{Q}}=\lambda(\widehat{\theta}),\textbf{R})~~~~~~~~~\\
=P( \widehat{\epsilon_j} > \rho \mid \epsilon_j=\widehat{\epsilon_j},\widehat{\textbf{Q}}=\lambda(\widehat{\theta}))P(\epsilon_j=\widehat{\epsilon_j}\mid\widehat{\textbf{Q}}=\lambda(\widehat{\theta}),\textbf{R}).
\end{align}
Therefore, the risk loss function can be determined by considering the posterior probability of an Eavesdropping event over the time-variant channel and is defined as follows:
\begin{align}
\gamma(\theta ,\epsilon_j)
=\delta_{Eve}(\widehat{\textbf{Q}})P(\widehat{\textbf{\textbf{Q}}}=\lambda(\widehat{\theta}),\textbf{V})+\delta_{var}(\widehat{\textbf{Q}})P(\widehat{\textbf{\textbf{Q}}}=\lambda(\widehat{\theta}),\textbf{R})\notag\\ -\delta_{Eve}(\widehat{\textbf{Q}})\delta_{var}(\widehat{\textbf{Q}})P(\widehat{\textbf{\textbf{Q}}}=\lambda(\widehat{\theta}),\textbf{V})\\
=\left [\delta_{Eve}(\widehat{\textbf{Q}})P(\textbf{V})+\delta_{var}(\widehat{\textbf{Q}})P(\textbf{R})-\delta_{Eve}(\widehat{\textbf{Q}})\delta_{var}(\widehat{\textbf{Q}})P(\textbf{V})\right ]\notag\\ \times P(\widehat{\textbf{\textbf{Q}}}=\lambda(\widehat{\theta})),
\end{align}
where $\lambda(\widehat{\theta})$ can be chosen as $GMM_{\widehat{\theta}}$  from our proposed approach, and the calculation of $P(\textbf{R})$ will be presented later in Section V. As shown in \cite{achab2020ranking}, the Empirical Risk Minimization (ERM) has an optimal parameter $\theta_{opt}$ using the best learner $\lambda_{opt}$ such that $R_{\varepsilon}(\lambda(\widehat{\theta}))>R_{\varepsilon}(\lambda_{opt}(\theta_{opt}))$. However, the presence of $\lambda_{opt}(\theta_{opt})$ is not always guaranteed, and the performance of $R_{\varepsilon}(\lambda(\widehat{\theta}))$ can be bound, as presented in \cite{Boucheron2010}.

In general, we further consider the case that the eavesdropping occurrence is unknown to the QKD networks. Following the example given in \cite{achab2020ranking}, the Bayes classifier $T_{Bayes}$ can be utilized by the empirical learner to detect instances of latent eavesdropping in QKD networks as follows:
\begin{equation}
    T_{Bayes}= \left\{\begin{matrix}
1,\eta(\epsilon ) > \alpha
\\ 
0,\eta(\epsilon ) < \alpha,
\end{matrix}\right.
\end{equation}
where $\alpha $ denotes the defense gate employed by Alice in order to quantify Eve's detection sensitivity and $\eta(\epsilon )$ denotes the posterior probability based on the proposed empirical learner, such as
\begin{align}
\eta(\epsilon)&=P\{\Delta, \epsilon=\widehat{\epsilon}\mid\widehat{\textbf{Q}}) \\
&=P(\Delta \mid \epsilon=\widehat{\epsilon},\widehat{\textbf{Q}})P(\epsilon=\widehat{\epsilon}\mid\widehat{\textbf{Q}})
\end{align}
\begin{align}
=P(\epsilon=\widehat{\epsilon}\mid\widehat{\textbf{Q}})\int_{\epsilon_{min}}^{1}\int_{\epsilon_{e}}^{1}f_{\widehat{\epsilon}} (\widehat{\epsilon}\mid\widehat{\textbf{Q}}=\lambda(\widehat{\theta}),\epsilon=\widehat{\epsilon},\epsilon_{e})\notag \\ \times f_{\epsilon_{e}}(\epsilon_{e})d\widehat{\epsilon}d\epsilon_{e},
\end{align}
where $f_{\widehat{\epsilon} }$ can be obtained by our proposed approach as GMM. The term $\epsilon_{e}$ is a random variable representing the minimum QBER incurred by Eve's interference, and the $\epsilon_{min}$ denotes the minimum QBER that can occur in the event of eavesdropping. The term 
$f_{\epsilon_{e}}(\epsilon_{e})$ can be assumed to have a uniform distribution of $[\epsilon_{min},\epsilon_{max}]$. In an extreme case of Trojan-horse attacks, the value of $\epsilon_{min}$ is 0.05 \cite{Jain2015}.
\begin{equation}
   \eta(\epsilon) =P(\epsilon=\widehat{\epsilon}\mid\widehat{\textbf{Q}})\int_{\epsilon_{min}}^{1}f_{\widehat{\epsilon}} (\widehat{\epsilon}\mid\widehat{\textbf{Q}}=\lambda(\widehat{\theta}),\epsilon=\widehat{\epsilon})d\widehat{\epsilon}. 
\end{equation}
Therefore, the optimal classification rule, as the Bayes classifier $T_{Bayes}$ is able to identify the difference between eavesdropping and time variation,
\begin{align}
P(\textbf{V})&=P(T_{Bayes}=1)\\
&=P(\eta(\epsilon ) > \alpha)\\
&=\frac{1}{N}\sum_{j=1}^{N} T_{Bayes}\{\eta(\epsilon_j ) > \alpha\}.
\end{align}

The value $\alpha$ will be discussed in Section V.
Substitute into the risk loss function as follows:
\begin{align}
\gamma(\widehat{\theta}, \epsilon_j) &= \left[ \delta_{Eve}(\widehat{\textbf{Q}})P(T_{Bayes}=1) 
+ \delta_{var}(\widehat{\textbf{Q}})P(\textbf{R}) \right. \notag \\
&\quad \left. - \delta_{Eve}(\widehat{\textbf{Q}})\delta_{var}(\widehat{\textbf{Q}})P(T_{Bayes}=1) \right] 
P(\widehat{\textbf{Q}}=\lambda(\widehat{\theta}))  \\
&= \left[ \delta_{Eve}(\widehat{\textbf{Q}})P(\eta(\epsilon) > \alpha) 
+ \delta_{var}(\widehat{\textbf{Q}})P(\textbf{R}) \right. \notag \\
&\quad \left. - \delta_{Eve}(\widehat{\textbf{Q}})\delta_{var}(\widehat{\textbf{Q}})P(\eta(\epsilon) > \alpha) \right] 
P(\widehat{\textbf{Q}}=\lambda(\widehat{\theta})).
\end{align}
In general, this theoretical analysis of risk evaluation is based on the accurate estimation of $\textbf{Q}$ such that $P(\epsilon=\widehat{\epsilon}\mid\widehat{\textbf{Q}})\approx1$. In order to investigate the merit of the proposed empirical risk minimization for QBER estimation, we propose an empirical learner using the generative model, not only using KS test results to fit the model but also adapting to the variant of the unsupervised dataset to guarantee a perfect model fitting by empirical data in the identical spirit of \cite{achab2020ranking}. The idea behind \cite{achab2020ranking} is to amend the difference in the empirical model between the training data and the test data by weighting risk. Our proposed learning algorithm adapts the training data to the test data by extending the categories of the generative model GMMs.
Therefore, the proposed risk measurement $\widehat{R_{\epsilon}}(\widehat{\theta})$ is defined as follows:
\begin{equation}
      \widehat{R_{\epsilon}}(\widehat{\theta}) = \frac{1}{N}\sum_{j=1}^{N}  H_{M_j}\gamma(\widehat{\theta} ,\epsilon_j ). 
\end{equation}
Since the proposed QBER estimation is able to provide an accurate estimation, apart from the approach presented in \cite{achab2020ranking}, the risk control weighting of $H^r_{M_j}$  can be given to suppress the function $\gamma$ in order to adapt the risk loss function. 
\begin{equation}
    H^r_{M_j}= 1-\gamma(\widehat{\theta} ,\epsilon_j ),~i=1...M_{ij}.
\end{equation}
Consequently, the proposed empirical risk reference for the QKD network is provided based on our proposed learner $\widehat{\textbf{Q}}$ as follows:
\begin{equation}
   \widehat{R_{ref}}(\widehat{\theta}) = \frac{1}{N}\sum_{j=1}^{N} (1-\gamma(\widehat{\theta} ,\epsilon_j ))\gamma(\widehat{\theta} ,\epsilon_j ).
\end{equation}
Therefore, the risk reduction rate $\beta_{\varepsilon}$  on each measured QBER can be derived and compared with the equal risk weighting $H(\kappa_{ij})=1$ as follows:
\begin{equation}
    \beta_{\varepsilon} = \gamma(\widehat{\theta} ,\epsilon_j ) \times100\%.
\end{equation}
The maximum risk reference can be a value of 0.25, while $\gamma(\widehat{\theta} ,\epsilon_j )=0.5$ and the minimum risk reduction rate is equal to 50\%, where $0<\gamma(\widehat{\theta} ,\epsilon_j )<1$.

Finally, the aforementioned empirical risk measurement of  $\widehat{R_{\varepsilon}}$ can be referred to as a risk reference $\widehat{R_{ref}}(\widehat{\theta})$ for all of the QKD vendors that provide customers with hourly-updated risk-aware monitoring based on the following conditions. It is worth noting that the difference between $\widehat{R_{\epsilon}} (\widehat{\theta})$ and $\widehat{R_{ref}}(\widehat{\theta})$ is that $H(\cdot)$ of $\widehat{R_{\epsilon}} (\widehat{\theta})$ applies the risk control weighting and $H(\cdot)$ of $\widehat{R_{ref}}(\widehat{\theta})$ adapts to the risk of the QKD device environment. The trust condition of QKD networks is defined as follows:
\begin{equation}
    \widehat{R_{\epsilon}} (\widehat{\theta})\leq \widehat{R_{ref}}(\widehat{\theta}). 
\end{equation}
For the trusted QKD networks, this condition provides safety assurance to monitor and measure the borderline of key transmissions that are at high risk of potential Trojan-horse attacks, i.e., $\widehat{R_{\epsilon}} (\widehat{\theta})>\widehat{R_{ref}}(\widehat{\theta})$. 
\section{Category-based Goodness-of-Fit GMM learning of QBER estimation for risk measurement}
In this section, we provide a novel approach for unsupervised machine learning, specifically designed for the learner denoted as $\widehat{\textbf{Q}}=\lambda(\widehat{\theta})$, as described in Equation \ref{gamma}. The category-based GMM \cite{Yan2017} is the utilization of statistical tools as an integral component of a data-driven approach. This evaluation procedure entails the implementation of a method known as soft clustering, where data samples are assigned to distinct groups based on certain criteria, resulting in the generation of a numerical categorization output. As an alternative to the methodology presented in \cite{Yan2017}, we propose using unsupervised learning techniques to facilitate the model-fitting process. The selection of the Kolmogorov-Smirnov (KS) criteria is based on its status as the sole extensively established goodness-of-fit criterion that exhibits competitiveness when compared to other methods examined in the literature, particularly in relation to shift and comparable alternatives. Furthermore, the utilization of this method allows for the development of straightforward confidence processes and tests. The approach we propose combines the utilization of the two-sample KS test \cite{berger2014kolmogorov} to assess the similarity of the two samples in terms of their distribution. This technique is employed as a means of presenting an innovative methodology. This is achieved by evaluating the P-value \cite{moscovich2013} between empirical QBER data and data generated by the tentative GMM. In \cite{moscovich2013}, the resolution to a well-recognized constraint of conventional P-value lies in its limited ability to identify deviations occurring at the extreme ends of the distribution. This test is employed to evaluate the goodness of fit and determine the appropriate categorization for new data points when the dataset does not conform well to the present category of GMM cluster distribution. With the core of the Expectation Maximization (EM) algorithm, we present the proposed category-based GMM KS learning to provide the estimated QBER pdf for the aforementioned empirical risk analysis. 
\begin{algorithm}\label{alg1}
	\caption{GMM model fitting using EM KS-test}
	\begin{algorithmic}[1]
		\State \textbf{Input:} Target dataset $F_s$, GMM parameter $G_s$, Number of GMM clusters $2,\ldots,c_{max}$, GMM maximum trial number $T_{max}$, EM maximum iteration $I_{max}$ 
		\State \textbf{Output:} Distribution set of GMMs with parameter $\{G_s\}_c$ and  $\{p_s\}_c$ as P-value of KS test; 
		\State \textbf{Initialize:} GMM with random parameter $G_s$ 
		\For {$c=2,\ldots,c_{max}$}
				\State  Run EM algorithm to fit the target dataset $F_s$: 
				\State  Store $gmm(G_s, c) \leftarrow gmm(G_s, c, {m}')$ and $\{p_s\}_c \leftarrow KS~test(gmm_s(G_s, c, m'),F_s)$, where $m'=\argmaxA_{m}\{ p_s \leftarrow KS~test(gmm_s(G_s, c, m),F_s)\}, m= 1,.., T_{max}$;
		\EndFor       
		\State Store the GMM parameter $\{G_s\}_c$ of $ gmm(G_s, c)$, $c\in [2,c_{max} ] $;
		\State return $\{G_s\}_c $, $\{p_s\}_c$
        \end{algorithmic}
\end{algorithm}

As presented in Algorithm 1, we apply this learning methodology as the core of the proposed GMM KS learner $\lambda(\widehat{\theta})$. In line 2, the parameter $\widehat{\theta}$ can be obtained as $\{G_s\}_c$. In line 6, an exhaustive search of $T_{max}\times I_{max}$ is performed to fit the empirical data by comparing the similarity between $gmm_s(G_s, c, m)$ and $F_s$. The GMM parameter $\{G_s\}_c$ is recorded in line 7 along with its corresponding maximum P-value $\{p_s\}_c$. Next, the training phase of the proposed leaner is presented in Algorithm 2 using the core of Algorithm 1.  
\begin{algorithm}\label{alg2}
	\caption{Category-based GMM-KS-test learning in training phase}
	\begin{algorithmic}[1]
		\State \textbf{Input:} Divide the training dataset into folds $F_1,...,F_K$, threshold $\varsigma$ to check the P-value for the KS test 
		\State \textbf{Output:} Category dataset of $C_1,\ldots,C_H$, $\{\{G_s\}_c\}_h$ and $\{\{p_s\}_c\}_h, h=1,\ldots,H$
		\State \textbf{Initialize:} $C_1\leftarrow F_1, h=1$ used as the input of Algorithm 1
		\For {$s=2,\ldots,K$}
				\State  Run Algorithm 1 to fit the target dataset $F_s$ and set $T_{max}=T_{Training}$:
                    \State ~~~~output $\{G_s\}_c \rightarrow\{\{G_s\}_c\}_h $ and $\{p_s\}_c\rightarrow\{\{p_s\}_c\}_h$  
                    \If{$\{p_s\}_c>\varsigma$}{
                     $ F_s \rightarrow C_h $}
                    \Else \State $ F_s \rightarrow C_{h+1} $ augment a new category $C_{h+1}$ \State $h=h+1$
                    \EndIf
		\EndFor         
		\State return $C_h$, $\{\{G_s\}_c\}_h $ and $\{\{p_s\}_c\}_h$ 
        \end{algorithmic}
\end{algorithm}

In line 1, the empirical training dataset can be partitioned into many folds. In order to establish a suitable correspondence between the two samples, a threshold $\varsigma$ is provided. As indicated in line 7, the purpose of this threshold is to maximize the probability $P(\epsilon=\widehat{\epsilon}\mid\widehat{\textbf{Q}})$ while minimizing the number of categories. Upon the failure of Algorithm 1 to adequately maintain the desired level of fit, a new category is established in line 9. Finally, the testing phase of Algorithm 3 involves conducting identical learning as Algorithm 2 and is stated as follows.

In lines 6 to 10, the sole distinction lies in the absence of enhanced categories during the test phase, and instead, the focus is on identifying the most suitable match within the pre-existing categories. Following the completion of the learning process, the learner performs a thorough and comprehensive search of the $T_{Test}$ for each category, as indicated in lines 12 to 14.
\begin{algorithm}\label{alg3}
	\caption{Category-based GMM-KS-test learning during the testing phase}
	\begin{algorithmic}[1]
		\State \textbf{Input:} Divide the testing dataset into folds $F_1,...,F_Z$, threshold $\varsigma$ to check the P-value for the KS test, and category dataset of $C_1,...,C_H$ from the training phase
		\State \textbf{Output:}  $\{\{G_s\}_c\}_h$ and $\{\{p_s\}_c\}_h, h=1,...,H$
		\State \textbf{Initialize:} $\{\{G_s\}_c\}_h $ using as the initial parameter of Algorithm 1
		\For {$s=1,\ldots,Z$}
                \For {$h=1,\ldots,H$}
				\State  Run Algorithm 1 using the parameter $\{\{G_s\}_c\}_h $ to fit the target dataset $F_s$ with $T_{max}=T_{Training}$:
                    \State ~~~~output $\{\{G_s\}_c\}_h $ and $\{\{p_s\}_c\}_h$   
                    \If{$\{\{p_s\}_c\}_h>\varsigma$ or h=H}{
                     $ F_s \rightarrow C_h $} break
                    \EndIf
                \EndFor     
		\EndFor    
             \For {h=1\ldots,H}
				\State  Run Algorithm 1 using the parameter $\{\{G_s\}_c\}_h $ and the target dataset $C_h$ with $T_{max}=T_{Test}$:
                    \State ~~~~output $\{\{G_s\}_c\}_h $ and $\{\{p_s\}_c\}_h$  
                \EndFor  
		\State return $\{\{G_s\}_c\}_h $ and $\{\{p_s\}_c\}_h$ 
        \end{algorithmic}
\end{algorithm}
\section{ Performance analysis of the proposed risk-aware machine learning}
The category-based goodness-of-fit GMM learning provides a learner with a model fitting the QBER model. Therefore, the probability $P\{\widehat{\textbf{\textbf{Q}}}=\lambda(\widehat{\theta})\}$ can be empirically calculated in the following form using the learner $\lambda_i(\widehat{\theta})$.
\begin{equation}
P\{\widehat{\textbf{\textbf{Q}}}=\lambda(\widehat{\theta})\} =\sum_{i=1}^{H}  P(\widehat{\textbf{\textbf{Q}}}\mid \lambda_i(\widehat{\theta}))P(\lambda_i(\widehat{\theta})),
\end{equation}
where $P(\lambda_i(\widehat{\theta}))=N_i/N$ and $N_i$ is the number of category sample.  

Therefore, the risk loss function can be derived as follows:
\begin{align}
\gamma(\widehat{\theta} ,\epsilon_j  ) = & \sum_{i=1}^{H}\left [\delta_{Eve}(\widehat{\textbf{Q}}\mid \lambda_i(\widehat{\theta}))P(\textbf{V}) + \delta_{var}(\widehat{\textbf{Q}}\mid \lambda_i(\widehat{\theta}))P(\textbf{R}) \right. \notag\\
&\left. - \delta_{Eve}(\widehat{\textbf{Q}}\mid \lambda_i(\widehat{\theta}))\delta_{var}(\widehat{\textbf{Q}}\mid \lambda_i(\widehat{\theta}))P(\textbf{V}) \right ] \notag\\
& \times P(\widehat{\textbf{Q}}\mid \lambda_i(\widehat{\theta}))P(\lambda_i(\widehat{\theta})),
\end{align}
where 
\begin{align}
    P(\textbf{V})
=\frac{1}{N}\sum_{i=1}^{H}\sum_{j=1}^{N_i} T_{Bayes}\{\eta(\epsilon_j \mid \lambda_i(\widehat{\theta})) > \alpha\},
\end{align}
and
\begin{equation}
    P(\textbf{R})
=\frac{1}{N}\sum_{i=1}^{H}\sum_{j=1}^{N_i} T_{Bayes}\{\eta(\epsilon_j \mid \lambda_i(\widehat{\theta})) > \alpha_{1m}\},
\end{equation}
where the value of $\alpha_{1m}$ is based on the posterior probability of $\delta_{var}(\widehat{\textbf{Q}}\mid \lambda_i(\widehat{\theta}))$ over a quantum channel that is 1 meter long. In our simulation, the parameter $\alpha_{1m}$ is assigned a constant value of 0.001\% for the quantum channels with lengths of 1km and 30km. In accordance with the principle of the Bayes classifier, the value of $\alpha$ is established by the maximal loss value of the GMM category. This value is dependent on the proposed method of effectively learning the time-variant quantum channel in the absence of Eve's attacks.
For a constant gate level, we have 
\begin{equation}
    \alpha = \argmaxA_{i=1\ldots H} \delta_{var}(\widehat{\textbf{Q}}\mid \lambda_i(\widehat{\theta}), \tilde{\textbf{V}}),
\end{equation}
where $\tilde{\textbf{V}}$ denotes the case of the absence of Eve's attacks. According to \cite{Devroye1996}, the optimal gate $\alpha^i_{opt}$ should depend on the GMM category as follows:
\begin{equation}
    \alpha^i_{opt} =  \delta_{var}(\widehat{\textbf{Q}}\mid \lambda_i(\widehat{\theta}), \tilde{\textbf{V}}).
\end{equation}
The proof of the Bayes optimal decision presented in \cite{Devroye1996} can be straightforwardly comprehended by substituting $P[Y=1\mid X=x]=\eta(\epsilon_j \mid \lambda_i(\widehat{\theta}))$ and $P[Y=0\mid X=x]=\alpha^i_{opt}$ for trusted QKD networks. Therefore, the Bayes classifier for identifying potential Eve attacks can be derived as follows:
\begin{equation}
    T_{Bayes}= \left\{\begin{matrix}
1,\eta(\epsilon_j \mid \lambda_i(\widehat{\theta})) > \alpha^i_{opt}
\\ 
0,\eta(\epsilon_j \mid \lambda_i(\widehat{\theta})) < \alpha^i_{opt}.
\end{matrix}\right.
\end{equation}
Finally, the proposed risk-aware technique of Sections III and IV is summarised and illustrated in Fig. \ref{fig_riskAI}. The upper blocks of the design computation can be viewed as the initialization of the proposed approach, which is necessary for the calculation of the defense gate of the Eve-free event. A trustworthy monitor is able to identify the attacks that occurred during Trojan Horse Eve's attack in the lower blocks. It is worth noting that $\alpha_{1m}$ can be obtained by calculating (9) over a quantum channel length of 1 meter and $ \lambda_i(\widehat{\theta})$ of Algorithm 2 and 3 is different for risk analysis to detect any latent Eve attacks both in the train and test phases. 
\begin{figure*}[htbp]
\centering
\captionsetup{justification=centering,margin=2cm}
\centerline{\includegraphics[width=\textwidth]{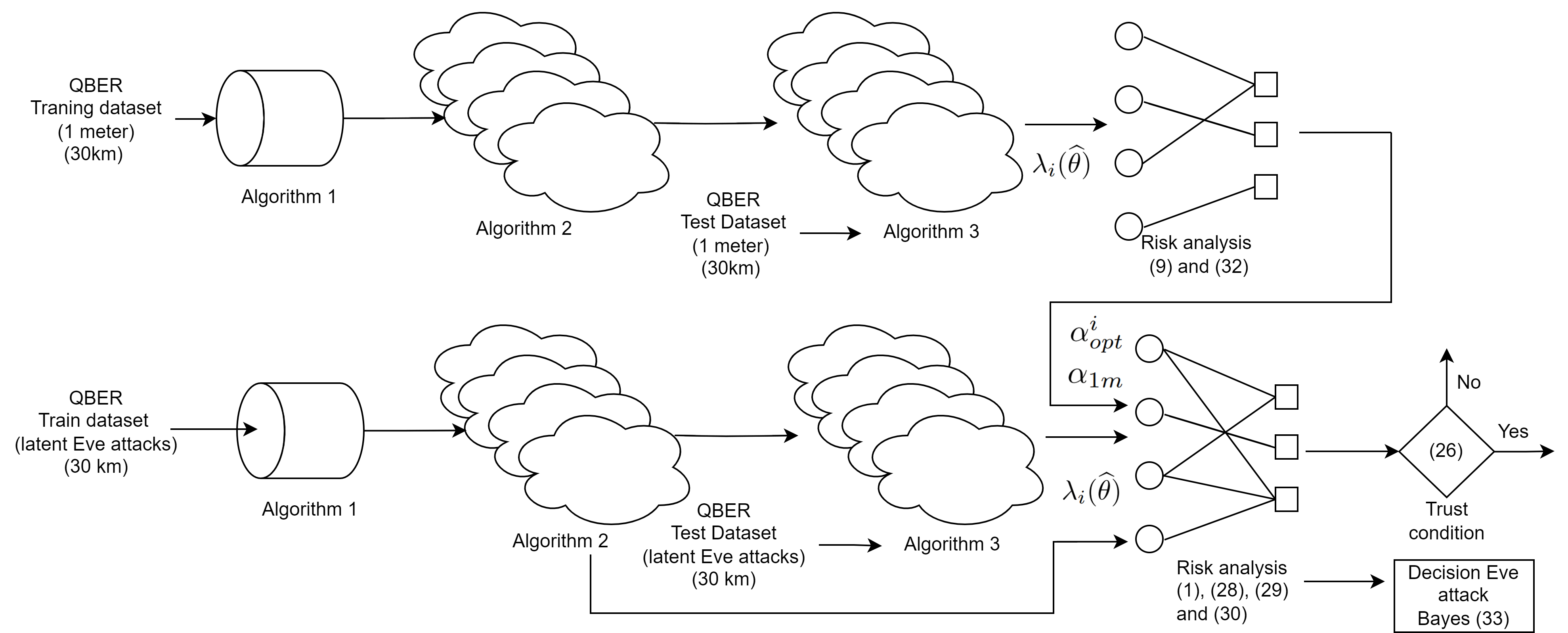}}
\caption{The overview of proposed risk-aware machine learning for latent Eve detection}
\label{fig_riskAI}
\end{figure*}
\begin{figure}[htbp]
\includegraphics[width=0.48\textwidth]{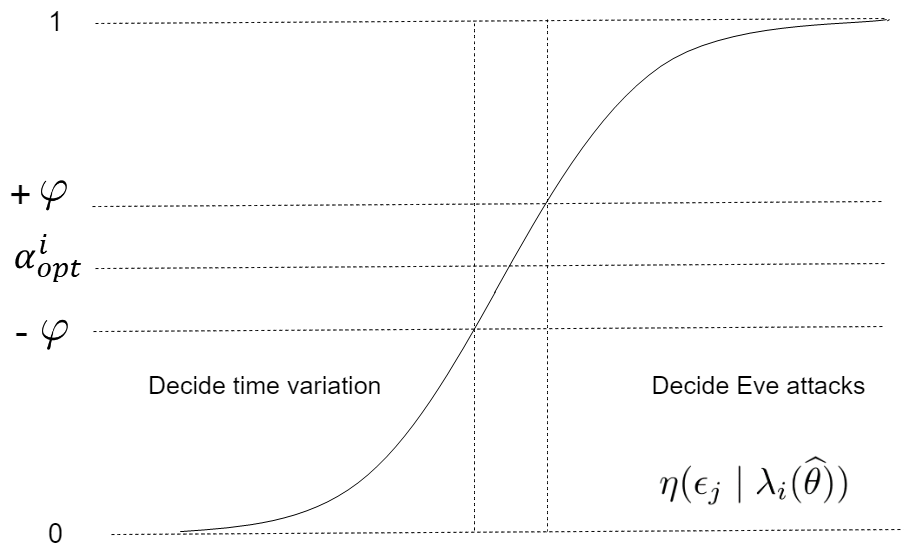}
\caption{The Bayes decision for trusted QKD scenario}
\label{fig1_2}
\end{figure}

In Section IV, the choice of $\varsigma=0.95$ guarantees that the P-value of the proposed learner possesses a $\varphi$ margin of estimated error with a corresponding occurrence probability of 5\%.  Furthermore, the erroneous detection probability is based on the calculation of the probability of the estimated error of the proposed learner. Therefore, the posterior probability $\tau$ that the Bayes classifier has an erroneous classification while considering the difference between $\eta(\epsilon_j \mid \lambda_i(\widehat{\theta}))$ and $\alpha^i_{opt}$. The difference could have an estimated error of less than $\varphi$ as illustrated in Fig. \ref{fig1_2} and apply the analogous of Markov Inequality as follows:
\begin{equation}
  \tau =P\{ \mid \eta(\epsilon_j \mid \lambda_i(\widehat{\theta})) - \alpha^i_{opt}\mid < \varphi \}
\end{equation}
\begin{align}
  =P\{  (\eta(\epsilon_j \mid \lambda_i(\widehat{\theta})) - \alpha^i_{opt} )< \varphi \}+\notag\\P\{  (\alpha^i_{opt} -\eta(\epsilon_j \mid \lambda_i(\widehat{\theta})))  < \varphi \}.
\end{align}
Therefore, 
\begin{align}
 \tau \leq \tau_{upper},~~~~~~~~~~~~~~~~~~~~~~~~~~~~~\notag\\
 \tau_{upper}=\frac{\textbf{E} ( (\eta(\epsilon_j \mid \lambda_i(\widehat{\theta})) - \alpha^i_{opt} ) ) }{\varphi}+ \frac{\textbf{E} (\alpha^i_{opt} -\eta(\epsilon_j \mid \lambda_i(\widehat{\theta}))) }{\varphi},
\end{align}
where the above 2 events occur with an equal probability of $(1-\varsigma)/2$ under the assumption of a uniform distribution and bound by the expectation value of $\textbf{E}$. We provide Lemma 1 and the proof as follows:
\begin{lemma}
    Given the QKD protocol with provable security, Alice and Bob explicitly detect the presence of Eve with high probability. As analogous to this trustworthy attribute, the Eve detection probability $\Psi$ is at least $1-\tau_{upper}(1-\varsigma)/2$ for our proposed trusted QKD scenario.
\end{lemma}
\begin{proof}
We apply the theorem of Markov inequality to calculate the upper bound of $\tau$ as follows:
Let $\phi^{+}$  and $\phi^{-}$ denote the event  $\{ (\eta(\epsilon_j \mid \lambda_i(\widehat{\theta})) - \alpha^i_{opt})< \varphi \}$ and $\{ (\alpha^i_{opt} -\eta(\epsilon_j \mid \lambda_i(\widehat{\theta})))< \varphi \}$ respectively. Let \textbf{A} and \textbf{B} denote the random variable of $(\eta(\epsilon_j \mid \lambda_i(\widehat{\theta})) - \alpha^i_{opt})$ and $(\alpha^i_{opt} -\eta(\epsilon_j \mid \lambda_i(\widehat{\theta})))$ respectively.
\begin{equation}
  \textbf{E}(\phi^{+})=\sum_{g \in G}p\{g\}\phi^{+}(g)= \sum_{g \in \phi^{+}}p\{g\}\textbf{A}(g)+ \sum_{g \in \bar{\phi^{+}}}p\{g\}\textbf{A}(g).
\end{equation}
Similarly, we have
\begin{equation}
  \textbf{E}(\phi^{-})=\sum_{g \in G}p\{g\}\phi^{-}(g)= \sum_{g \in \phi^{-}}p\{g\}\textbf{B}(g)+ \sum_{g \in \bar{\phi^{-}}}p\{g\}\textbf{B}(g).
\end{equation}
Since $\textbf{A}$ and $\textbf{B}$ are positive value,
\begin{equation}
   \textbf{E}(\phi^{+})\geq  \sum_{g \in \phi^{+}}p\{g\}\textbf{A}(g)\geq \varphi \sum_{g \in \phi^{+}}p(g)=\varphi P\{\textbf{A}<\varphi\},
\end{equation}
and
\begin{equation}
   \textbf{E}(\phi^{-})\geq  \sum_{g \in \phi^{-}}p\{g\}\textbf{B}(g)\geq \varphi \sum_{g \in \phi^{-}}p(g)=\varphi P\{\textbf{B}<\varphi\},
\end{equation}
Therefore,
\begin{equation}
  \tau_{upper}=(\textbf{E}(\phi^{+}) +\textbf{E}(\phi^{-}))/\varphi \geq P(A<\varphi) +P(B<\varphi)=\tau
\end{equation}
Consequently, 
\begin{equation}
    \Psi \geq 1-\tau_{upper}(1-\varsigma)/2.
\end{equation}
\end{proof}
Given Lemma 1 and the high probability of Eve detection as the value of $1-\tau(1-\varsigma)/2$, it can be extrapolated that the proposed trusted QKD scenario attains trustworthy QKD networks.

\section{Numerical result}
\subsection{Trusted QKD Scenario Configuration}
The statistical data for the QBER was acquired via the LUQCIA project\footnote{https://www.uni.lu/en/news/a-first-testbed-for-quantum-communication-infrastructure-in-luxembourg/}. This data was gathered using the QNET WEBAPI interface version 1.168 at constant intervals over a span of many months. These data points were used to establish the sample space for our observations. 
\begin{itemize}
    \item The experimental setup for measuring quantum channel distance consists of three different distances: 1m, 1km, and 30km. The number of QBER experiment $N=\{N_{1m}, N_{1km},N_{30km} \}=\{57471,52104,47768 \}$.
    \item In the context of IDQ QKD devices, if the visibility is below 0.9 or the QBER exceeds the abortion rate of 11\%, it can be shown that Bob is unlikely to receive any raw keys reduced by privacy amplification, leading to a key rate of 0. Eve holds the belief that eavesdropping does not have the capability of impacting the visibility of data transmission, which is only influenced by the presence of an optical fiber link. The various distance quantum channel experiments are able to obtain the statistical QBER data via IDQ QKD device pairs. 
    \item In the proposed learning scenario, the values of $T_{Training}$ and $T_{Test}$ are determined as 100 and 10000, respectively. The value of $c_{max}$ is selected as 15 and 45. The maximum number of iterations $I_{max}$ for the EM algorithm is set to 100.
\end{itemize}

\subsection{QBER Estimation and Risk Analysis for a Trustworthy Monitor}
The P-value in Algorithm 2 is shown in Fig. \ref{fig1}, where the value of K is set to 5, 6, 7, and 8, and $N$ is set to $N_{1m}$. This presentation illustrates the application of GMM fitting to the empirical QBER data. In the test phase, cross-validation \cite{schaffer1993selecting} is a highly effective method for evaluating the performance of the suggested learner. The QBER empirical sample $N=\{N_{1km},N_{30km}\}$ is partitioned into four folds. In each cross-validation iteration, one fold is designated as the test set, while the remaining folds serve as the training set as illustrated by the P-value in Fig. \ref{fig2} over 1 kilometer and in Fig. \ref{fig4} over the 30km optical quantum channel. As shown in Fig. \ref{fig3} using Algorithm 3, the Akaike Information Criterion (AIC) is utilized in order to assess the performance of statistical models and ascertain their level of efficiency when applied to a certain dataset. Furthermore, we can observe the results of high-dimensional GMM fitting to the test dataset in Fig. \ref{fig5}, and an identical trend was observed for the AIC result. The goodness-of-fit of the high-dimensional GMM with a cluster range from 20 to 45 is better than the range from 2 to 15.

The risk analysis is based on the assumption that the risk control weighting $H_{M_j}$ follows a normal distribution. In order to demonstrate the effectiveness of the proposed approach, risk control weighting has a mean value that is uniformly distributed between 0.5 and 1, with a standard deviation that is equal to the mean value. Over a time-variant quantum channel, this configuration guarantees the simulation's inclusion, as illustrated in Fig. \ref{fig7}, Fig. \ref{fig8}, and in Fig. \ref{fig8_1} in accordance with the Trojan-horse challenge raised in \cite{Jain2015}. It is important to highlight that a high-dimensional GMM can provide better goodness-of-fit for the QBER when the risk occurrence $\epsilon_j$ exceeds the threshold $\rho$ and when there is an opportunity for identifying larger potential risks.

Finally, we implement Trojan-horse attacks on our trusted QKD scenario, wherein the QKD system is subjected to Eve's attacks following a Poisson distribution using parameter $\upsilon_e$. It is assumed that each instance of a Trojan horse attack results in an increase in the QBER to a uniform distribution ranging from 0.05 to 0.055. The observation of Alice's examination of the defense gate associated with Eve's arrival is intriguing. In Fig. \ref{fig9} and Fig. \ref{fig10}, the proposed trustworthy monitor can recognize the occurrence of Eve's Trojan horse attacks due to the violation of the trust condition. The value of $\alpha$ has been established as the value of $0.2\%$ according to the maximum value of the risk loss function as shown in Fig. \ref{fig8}. Additionally, Fig. \ref{fig10} sets $\upsilon_e$ to 4000, where the total number of Eve attacks is 10 times in order to compare with 100 times as shown in Fig. \ref{fig9}. It is noteworthy to observe that the risk analysis of Fig. \ref{fig9} is comparatively greater than Fig. \ref{fig10} and the proposed QKD scenario can detect Trojan-horse attacks with only the proposition of $0.02\%$. 

This discrepancy signifies a distinct sensitivity level in Eve's detection, demonstrating our proposed trustworthy QKD scenario by means of the trust condition. It is interesting to note that the defense gate $0.25\%$ is too high and less sensitive to detect potential Eve's attacks, comparing the risk analysis of Fig. \ref{fig9} and Fig. \ref{fig10}. Consequently, the potential Trojan horse Eve's attack can be successfully detected by our proposed trusted QKD scenario, as shown in the risk difference and loss difference of value between $0.02\%$ and $0.05\%$ in Fig. \ref{fig8} and Fig. \ref{fig9}. The design of the trustworthy monitor being proposed is predominately determined by the level of defense gate that is portrayed as obstructing Eve's attacks. In order to preserve generality, we presume that Eve is not acquainted with the time-variant channel effect.

\begin{figure}[htbp]
\centering
\centerline{\includegraphics[width=0.5\textwidth]{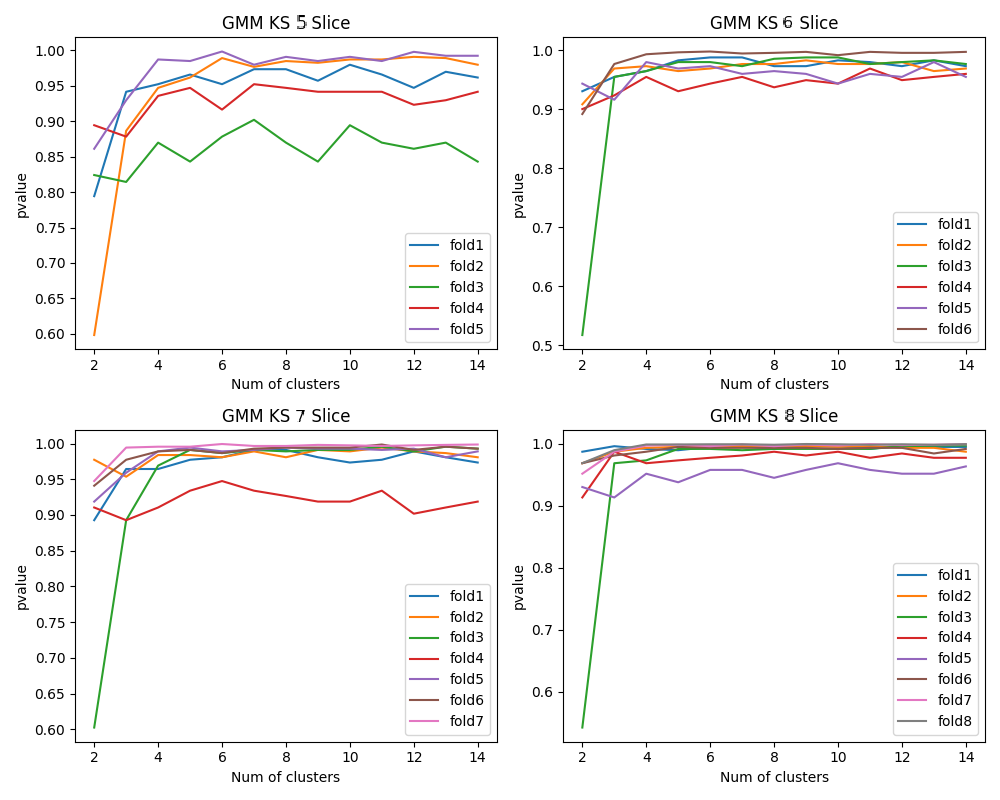}}
\caption{Multiple folders model fitting using Algorithm 2 over the 1m optical quantum channel}
\label{fig1}
\end{figure}
\begin{figure}[htbp]
\centerline{\includegraphics[width=0.5\textwidth]{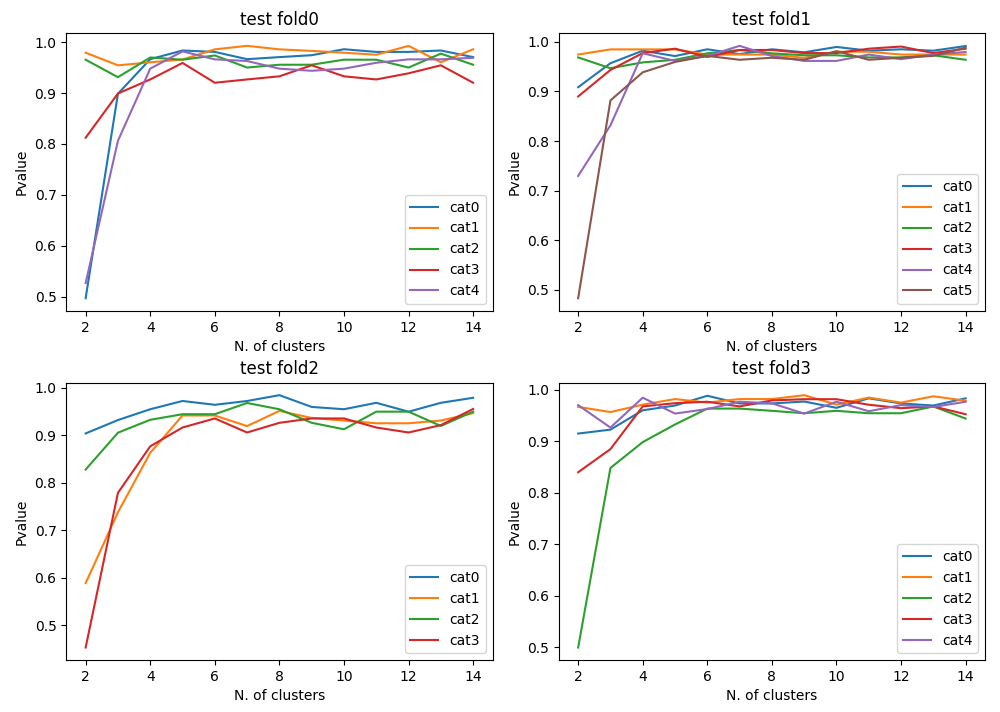}}
\caption{4-fold cross-validation for Algorithm 3 where $\varsigma=0.95$ over the 1km optical quantum channel}
\label{fig2}
\end{figure}
\begin{figure}[htbp]
\centerline{\includegraphics[width=0.5\textwidth]{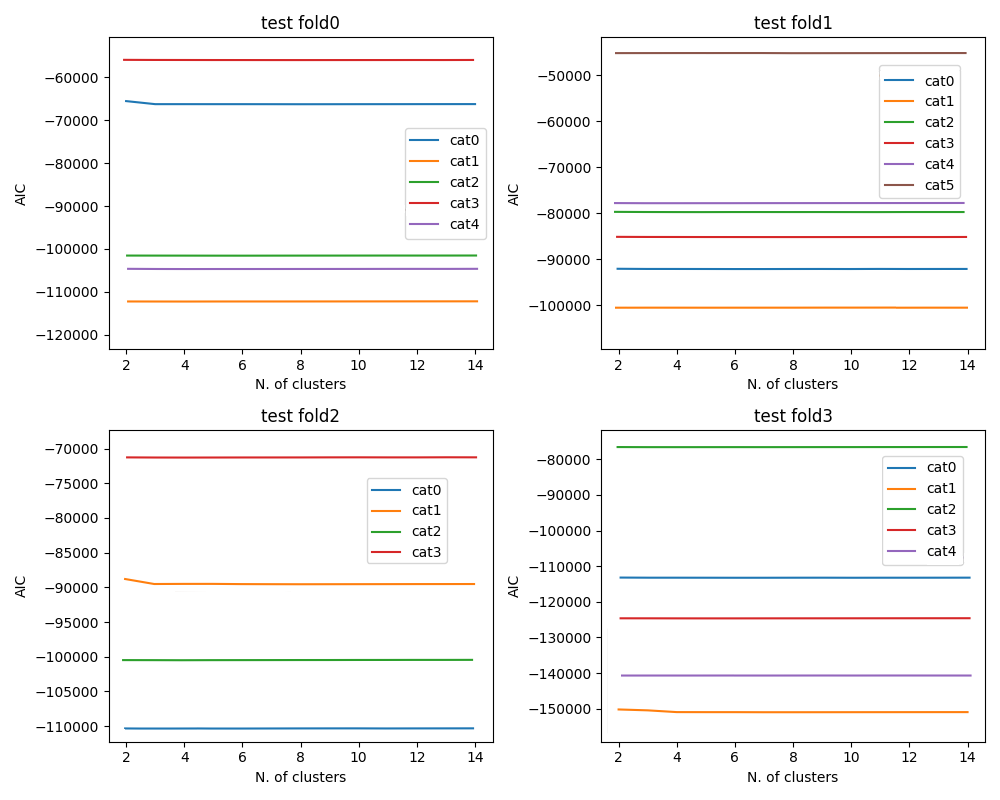}}
\caption{AIC Results of Algorithm 3 where $\varsigma=0.95$ over the 1km optical quantum channel}
\label{fig3}
\end{figure}
\begin{figure}[htbp]
\centerline{\includegraphics[width=0.5\textwidth]{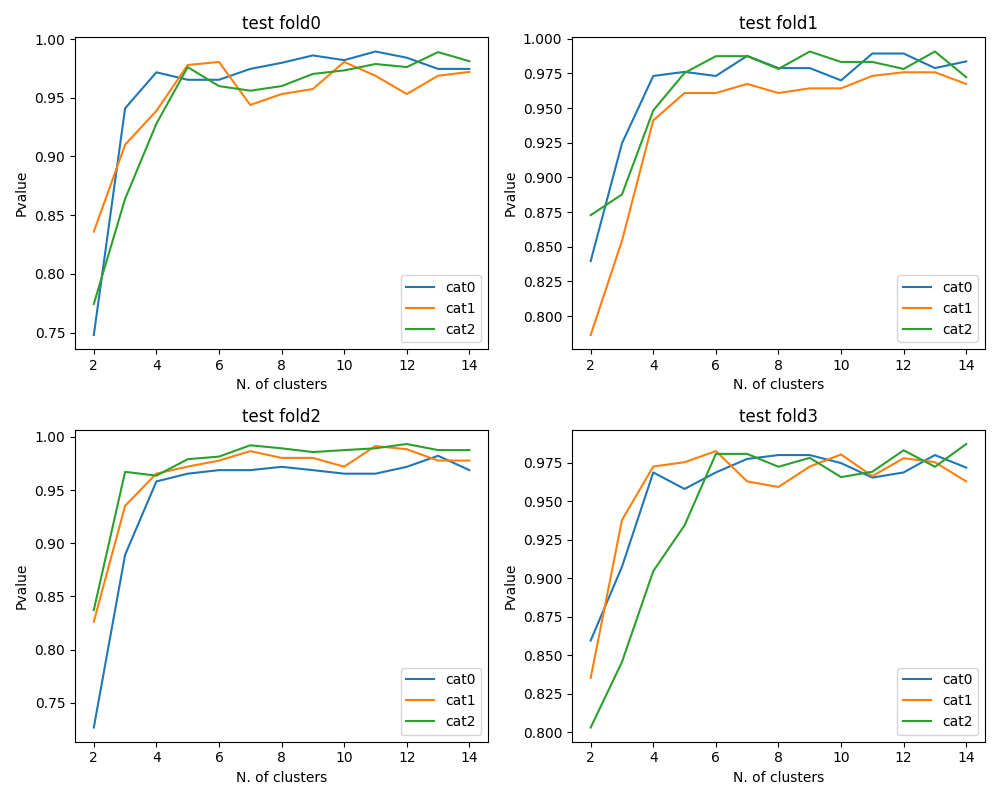}}
\caption{4-fold cross-validation for Algorithm 3 where $\varsigma=0.95$ over the 30km optical quantum channel}
\label{fig4}
\end{figure}
\begin{figure}[htbp]
\centerline{\includegraphics[width=0.5\textwidth]{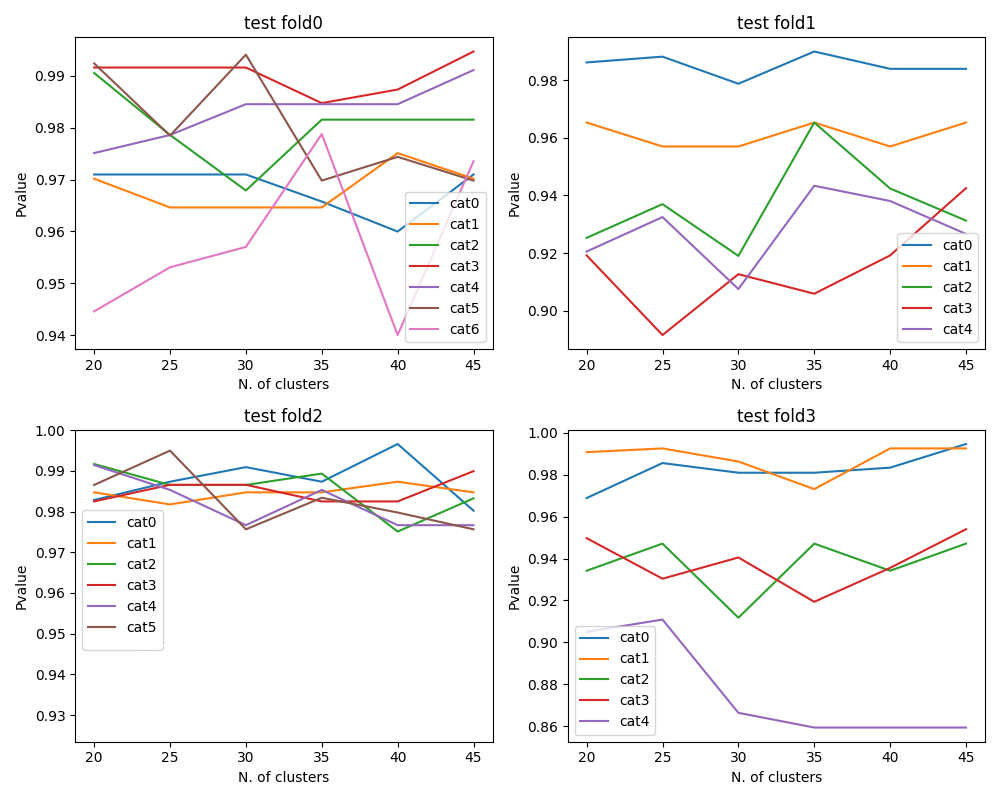}}
\caption{High dimension 4-fold cross-validation of Algorithm 3 where $\varsigma=0.95$ over the 1km optical quantum channel}
\label{fig5}
\end{figure}
\begin{figure}[htbp]
\centerline{\includegraphics[width=0.5\textwidth]{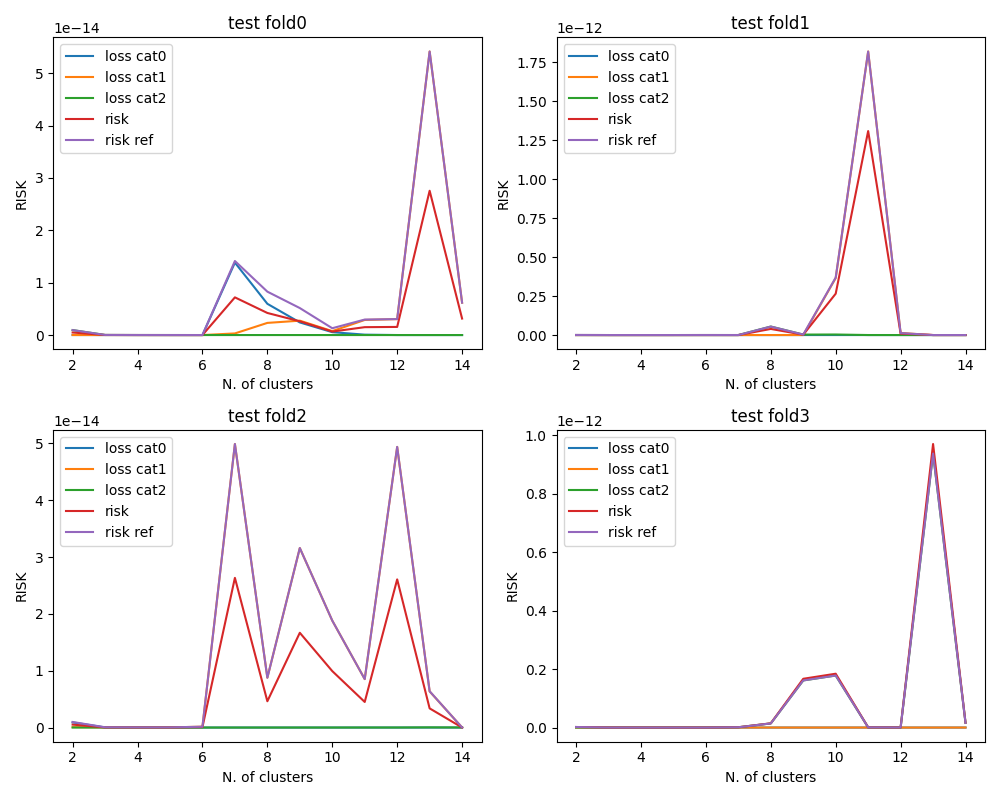}}
\caption{Risk analysis of cross-validation using Algorithm 3 where $\varsigma=0.95$ and $\rho=0.08$ over the 30km optical quantum channel}
\label{fig7}
\end{figure}
\begin{figure}[htbp]
\centerline{\includegraphics[width=0.5\textwidth]{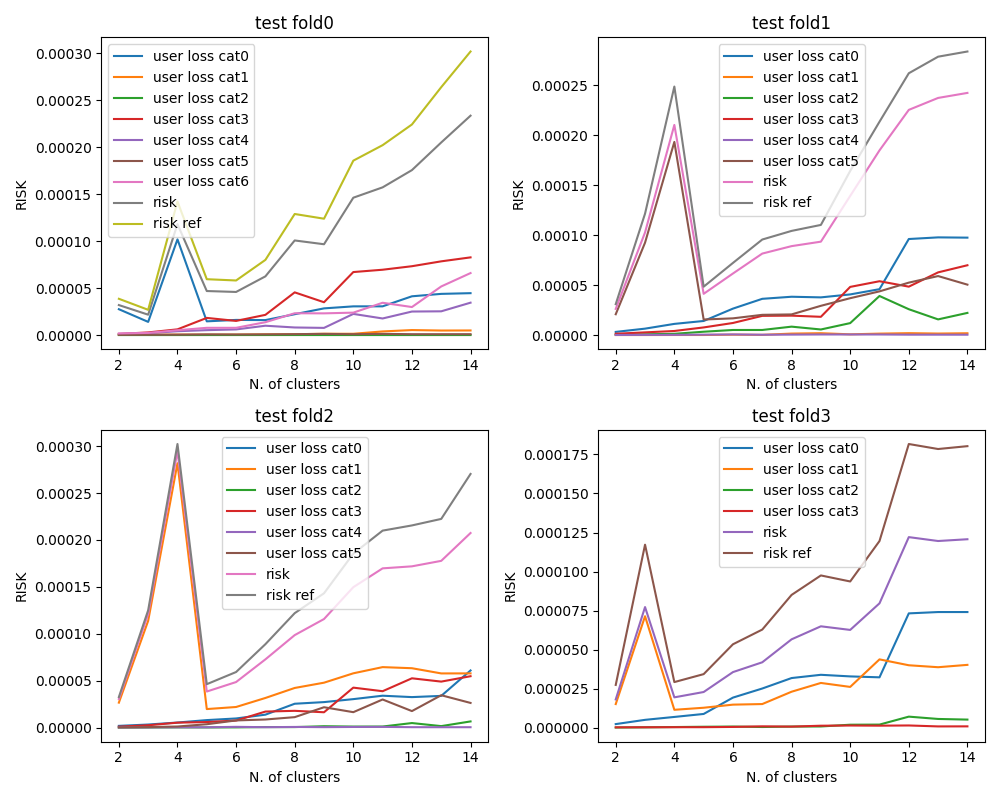}}
\caption{Risk analysis of cross-validation using Algorithm 3 where $\varsigma=0.95$ and $\rho=0.05$ over the 1km optical quantum channel}
\label{fig8_1}
\end{figure}
\begin{figure}[htbp]
\centerline{\includegraphics[width=0.5\textwidth]{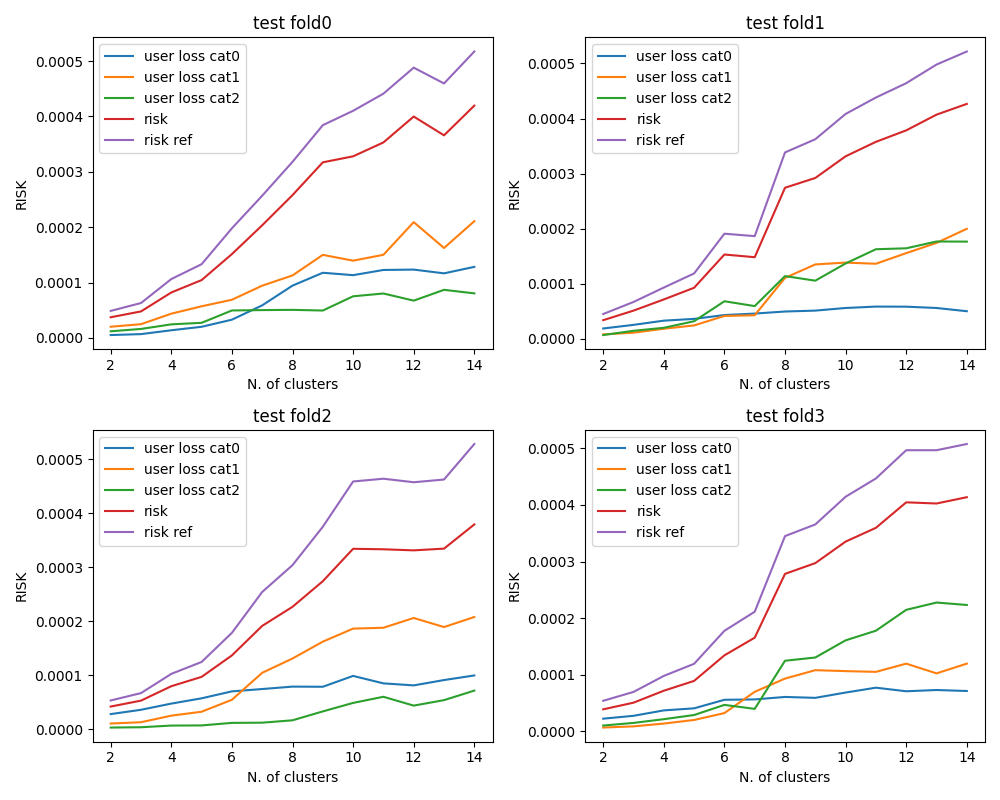}}
\caption{Risk analysis of cross-validation using Algorithm 3 where $\varsigma=0.95$ and $\rho=0.05$ over the 30km optical quantum channel}
\label{fig8}
\end{figure}
\begin{figure}[htbp]
\centerline{\includegraphics[width=0.5\textwidth]{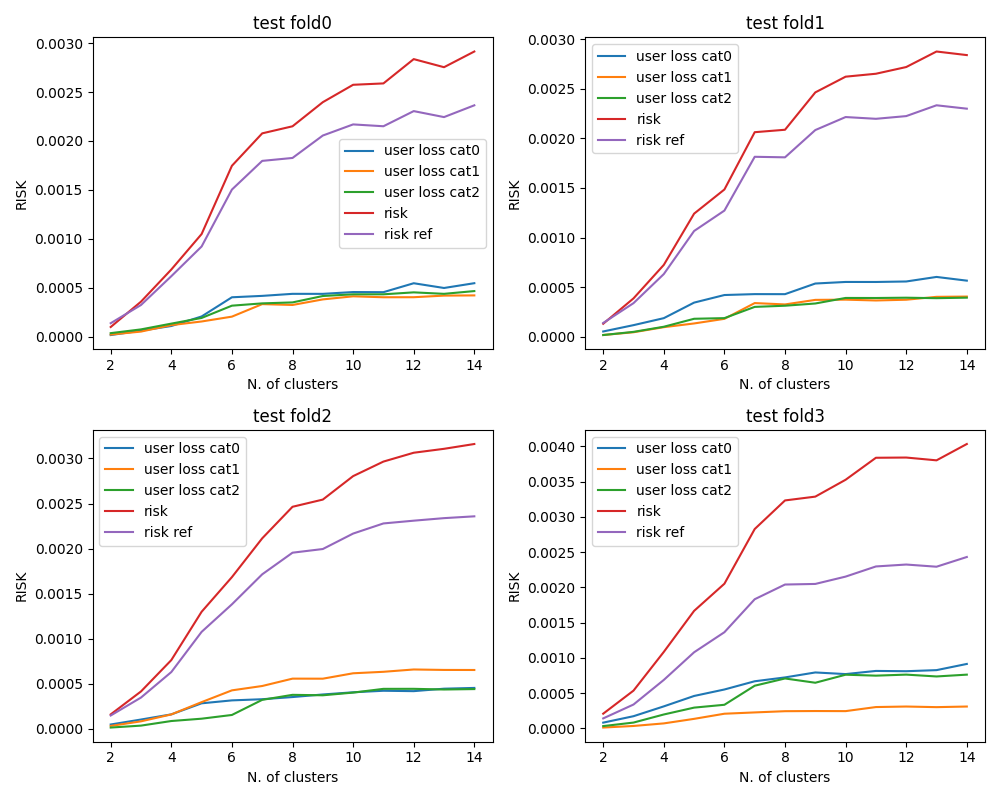}}
\caption{Risk analysis of cross-validation using Algorithm 3 where $\varsigma=0.95$, $\rho=0.05$, $ \alpha=0.2\%$ and $\upsilon_e=500$ with Eve's Trojan-horse attacks over the 30km optical quantum channel}
\label{fig9}
\end{figure}
\begin{figure}[htbp]
\centerline{\includegraphics[width=0.5\textwidth]{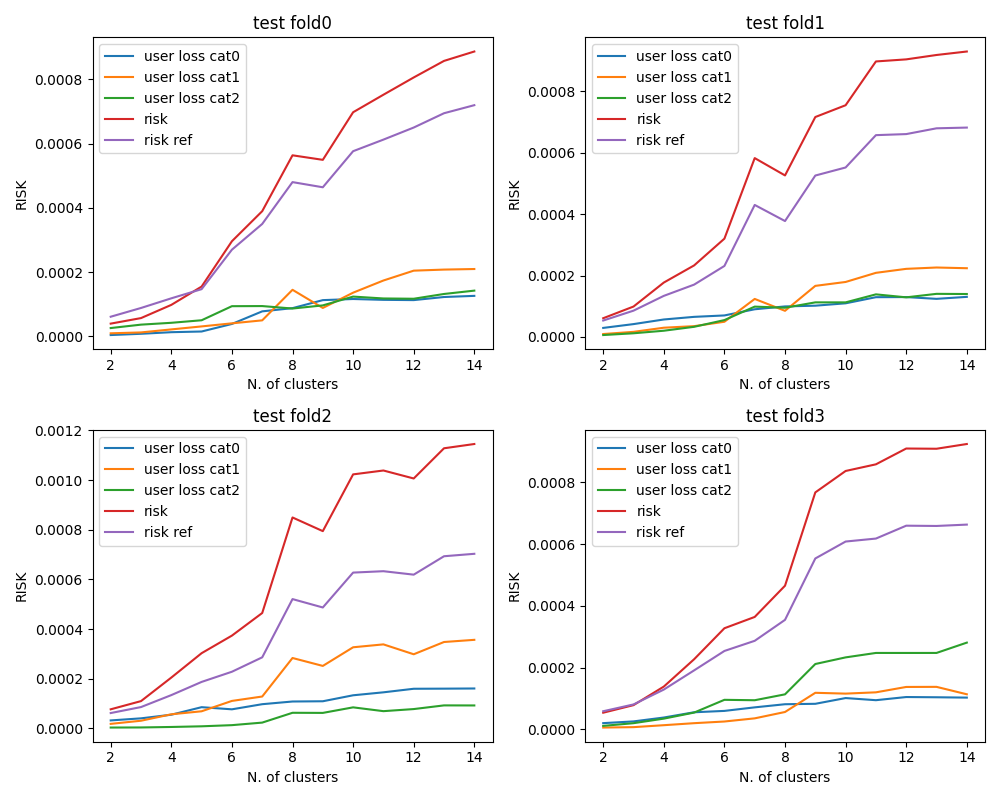}}
\caption{Risk analysis of cross-validation using Algorithm 3 where $\varsigma=0.95$, $\rho=0.05$, $\alpha=0.2\%$ and $\upsilon_e=4000$ with Eve's Trojan-horse attacks over the 30km optical quantum channel}
\label{fig10}
\end{figure}
\section{Conclusion and Future Directions}
The novelty of this study is to examine the incorporation of risk-aware machine learning methods into trustworthy QKD networks, with a particular emphasis on the vulnerability of credentials over a time-variant quantum channel. To accomplish this task, we employ IDQ QKD devices as a means of support. The empirical QBER dataset is effectively estimated by the proposed GMM KS learner, as indicated by the numerical results. The QKD device developed by IDQ is at the forefront of technology and meets the trust condition, making it trustworthy for ultra-secure communication. The simulation has been particularly designed to cater to the network traffic of the SDN controller, and inspired by an optimal classifier to effectively identify any conceivable Trojan horse attacks. The probability of detecting Eve in our suggested trustworthy QKD scenario is computed to demonstrate a significant level of confidence. 

In terms of future directions, this work encourages the adoption of a risk-aware reinforcement learning methodology. This methodology incorporates risk measurement and risk reference to formulate the reward function while considering the trust condition. The agents have the capability of utilizing the risk function recommended by the reward function, as well as online detection of Eve's arrivals and time-varying interference. The aforementioned future avenues can be pursued by drawing inspiration from the fields of Bayesian online change-point detection and risk-aware reinforcement learning. This enables the trusted QKD networks to gain valuable knowledge on the optimal trust policy. This is conducted within the framework of the trusted node and the variant network architecture for the cost-effective deployment of QKD networks.

\bibliographystyle{IEEEtran}
\bibliography{main}

\begin{thebibliography}{10}
\providecommand{\url}[1]{#1}
\csname url@samestyle\endcsname
\providecommand{\newblock}{\relax}
\providecommand{\bibinfo}[2]{#2}
\providecommand{\BIBentrySTDinterwordspacing}{\spaceskip=0pt\relax}
\providecommand{\BIBentryALTinterwordstretchfactor}{4}
\providecommand{\BIBentryALTinterwordspacing}{\spaceskip=\fontdimen2\font plus
\BIBentryALTinterwordstretchfactor\fontdimen3\font minus \fontdimen4\font\relax}
\providecommand{\BIBforeignlanguage}[2]{{%
\expandafter\ifx\csname l@#1\endcsname\relax
\typeout{** WARNING: IEEEtran.bst: No hyphenation pattern has been}%
\typeout{** loaded for the language `#1'. Using the pattern for}%
\typeout{** the default language instead.}%
\else
\language=\csname l@#1\endcsname
\fi
#2}}
\providecommand{\BIBdecl}{\relax}
\BIBdecl

\bibitem{Zio2016}
\BIBentryALTinterwordspacing
E.~Zio, ``Challenges in the vulnerability and risk analysis of critical infrastructures,'' \emph{Reliability Engineering and System Safety}, vol. 152, pp. 137--150, 2016. [Online]. Available: \url{https://www.sciencedirect.com/science/article/pii/S0951832016000508}
\BIBentrySTDinterwordspacing

\bibitem{NAIR2014}
\BIBentryALTinterwordspacing
S.~Nair, J.~L. {de la Vara}, M.~Sabetzadeh, and L.~Briand, ``An extended systematic literature review on provision of evidence for safety certification,'' \emph{Information and Software Technology}, vol.~56, no.~7, pp. 689--717, 2014. [Online]. Available: \url{https://www.sciencedirect.com/science/article/pii/S0950584914000603}
\BIBentrySTDinterwordspacing

\bibitem{Bolbot2019}
\BIBentryALTinterwordspacing
V.~Bolbot, G.~Theotokatos, L.~M. Bujorianu, E.~Boulougouris, and D.~Vassalos, ``{Vulnerabilities and safety assurance methods in Cyber-Physical Systems: A comprehensive review},'' \emph{Reliability Engineering and System Safety}, vol. 182, pp. 179--193, 2019. [Online]. Available: \url{https://www.sciencedirect.com/science/article/pii/S0951832018302709}
\BIBentrySTDinterwordspacing

\bibitem{bell2006introduction}
R.~Bell, ``Introduction to iec 61508,'' in \emph{{Proceedings of the 10th Australian workshop on Safety critical systems and software-Volume 55}}.\hskip 1em plus 0.5em minus 0.4em\relax Citeseer, 2006, pp. 3--12.

\bibitem{leveson2014comparison}
N.~Leveson, C.~Wilkinson, C.~Fleming, J.~Thomas, and I.~Tracy, ``{A comparison of STPA and the ARP 4761 safety assessment process},'' \emph{Massachusetts Institute of Technology, Cambridge, MA}, 2014.

\bibitem{feng2017trusted}
\BIBentryALTinterwordspacing
F.~D. and P.~T.U., \emph{{Trusted Computing: Principles and Applications}}, ser. Advances in Computer Science.\hskip 1em plus 0.5em minus 0.4em\relax De Gruyter, 2017. [Online]. Available: \url{https://books.google.lu/books?id=GLxGDwAAQBAJ}
\BIBentrySTDinterwordspacing

\bibitem{Parny_2023}
\BIBentryALTinterwordspacing
L.~de~Forges~de Parny, O.~Alibart, J.~Debaud, S.~Gressani, A.~Lagarrigue, A.~Martin, A.~Metrat, M.~Schiavon, T.~Troisi, E.~Diamanti, P.~G{\'{e}}lard, E.~Kerstel, S.~Tanzilli, and M.~V.~D. Bossche, ``Satellite-based quantum information networks: use cases, architecture, and roadmap,'' \emph{Communications Physics}, vol.~6, no.~1, jan 2023. [Online]. Available: \url{https://doi.org/10.1038%2Fs42005-022-01123-7}
\BIBentrySTDinterwordspacing

\bibitem{Trinh2018}
T.~P. V., P.~A. T., C.-C. Alberto, and T.~Moria, ``{Quantum Key Distribution over FSO: Current Development and Future Perspectives},'' in \emph{2018 Progress in Electromagnetics Research Symposium (PIERS-Toyama)}, 2018, pp. 1672--1679.

\bibitem{QKD1}
\BIBentryALTinterwordspacing
N.~Gisin, G.~Ribordy, W.~Tittel, and H.~Zbinden, ``Quantum cryptography,'' \emph{Rev. Mod. Phys.}, vol.~74, pp. 145--195, Mar 2002. [Online]. Available: \url{https://link.aps.org/doi/10.1103/RevModPhys.74.145}
\BIBentrySTDinterwordspacing

\bibitem{QKD2}
D.~Mayers, ``Unconditional security in quantum cryptography,'' 2004.

\bibitem{Lajos}
H.~Nedasadat, B.~Zunaira, M.~Robert, N.~S. Xin, and H.~Lajos, ``{Satellite-Based Continuous-Variable Quantum Communications: State-of-the-Art and a Predictive Outlook},'' \emph{IEEE Communications Surveys \& Tutorials}, vol.~21, no.~1, pp. 881--919, 2019.

\bibitem{Erven_2008}
\BIBentryALTinterwordspacing
C.~Erven, C.~Couteau, R.~Laflamme, and G.~Weihs, ``Entangled quantum key distribution over two free-space optical links,'' \emph{Optics Express}, vol.~16, no.~21, p. 16840, oct 2008. [Online]. Available: \url{https://doi.org/10.1364%2Foe.16.016840}
\BIBentrySTDinterwordspacing

\bibitem{Scheidl2009}
T.~Scheidl, R.~Ursin, A.~Fedrizzi, S.~Ramelow, X.~song Ma, T.~Herbst, R.~Prevedel, L.~Ratschbacher, J.~Kofler, T.~Jennewein, and A.~Zeilinger, ``Feasibility of 300km quantum key distribution with entangled states,'' \emph{New Journal of Physics}, vol.~11, p. 085002, 2009.

\bibitem{Fedrizzi_2009}
\BIBentryALTinterwordspacing
A.~Fedrizzi, R.~Ursin, T.~Herbst, M.~Nespoli, R.~Prevedel, T.~Scheidl, F.~Tiefenbacher, T.~Jennewein, and A.~Zeilinger, ``High-fidelity transmission of entanglement over a high-loss free-space channel,'' \emph{Nature Physics}, vol.~5, no.~6, pp. 389--392, may 2009. [Online]. Available: \url{https://doi.org/10.1038%2Fnphys1255}
\BIBentrySTDinterwordspacing

\bibitem{quantique2001quantis}
I.~Quantique and R.~Extractor, ``Quantis,'' \emph{Quantum number generator}, pp. 2001--2010, 2001.

\bibitem{chong2010quantum}
S.-K. Chong and T.~Hwang, ``Quantum key agreement protocol based on bb84,'' \emph{Optics Communications}, vol. 283, no.~6, pp. 1192--1195, 2010.

\bibitem{jain2016attacks}
N.~Jain, B.~Stiller, I.~Khan, D.~Elser, C.~Marquardt, and G.~Leuchs, ``Attacks on practical quantum key distribution systems (and how to prevent them),'' \emph{Contemporary Physics}, vol.~57, no.~3, pp. 366--387, 2016.

\bibitem{Jain2015}
N.~Jain, B.~Stiller, I.~Khan, V.~Makarov, C.~Marquardt, and G.~Leuchs, ``{Risk Analysis of Trojan-Horse Attacks on Practical Quantum Key Distribution Systems},'' \emph{IEEE Journal of Selected Topics in Quantum Electronics}, vol.~21, no.~3, pp. 168--177, 2015.

\bibitem{Naveenkumar2021}
R.~Naveenkumar, N.~Sivamangai, A.~Napolean, and V.~Janani, ``{A Survey on Recent Detection Methods of the Hardware Trojans},'' in \emph{2021 3rd International Conference on Signal Processing and Communication (ICPSC)}, 2021, pp. 139--143.

\bibitem{Zhang2021}
X.~Zhang, R.~Gupta, A.~Mian, N.~Rahnavard, and M.~Shah, ``Cassandra: Detecting trojaned networks from adversarial perturbations,'' \emph{IEEE Access}, vol.~9, pp. 135\,856--135\,867, 2021.

\bibitem{huang2020survey}
Z.~Huang, Q.~Wang, Y.~Chen, and X.~Jiang, ``{A survey on machine learning against hardware trojan attacks: Recent advances and challenges},'' \emph{IEEE Access}, vol.~8, pp. 10\,796--10\,826, 2020.

\bibitem{almeida2022ransomware}
F.~Almeida, M.~Imran, J.~Raik, and S.~Pagliarini, ``{Ransomware attack as hardware trojan: a feasibility and demonstration study},'' \emph{IEEE Access}, vol.~10, pp. 44\,827--44\,839, 2022.

\bibitem{yang2020dynamic}
D.~Yang, C.~Gao, and J.~Huang, ``{Dynamic game for strategy selection in hardware Trojan attack and defense},'' \emph{IEEE Access}, vol.~8, pp. 213\,094--213\,103, 2020.

\bibitem{sajeed2017invisible}
S.~Sajeed, C.~Minshull, N.~Jain, and V.~Makarov, ``{Invisible Trojan-horse attack},'' \emph{Scientific reports}, vol.~7, no.~1, p. 8403, 2017.

\bibitem{yang2015trojan}
Y.-G. Yang, S.-J. Sun, and Q.-Q. Zhao, ``Trojan-horse attacks on quantum key distribution with classical bob,'' \emph{Quantum Information Processing}, vol.~14, pp. 681--686, 2015.

\bibitem{Vinay2018}
\BIBentryALTinterwordspacing
S.~E. Vinay and P.~Kok, ``{Extended analysis of the Trojan-horse attack in quantum key distribution},'' \emph{Phys. Rev. A}, vol.~97, p. 042335, Apr 2018. [Online]. Available: \url{https://link.aps.org/doi/10.1103/PhysRevA.97.042335}
\BIBentrySTDinterwordspacing

\bibitem{Borisova2020RiskAO}
\BIBentryALTinterwordspacing
A.~V. Borisova, B.~D. Garmaev, I.~B. Bobrov, S.~Negodyaev, and I.~V. Sinil’shchikov, ``{Risk Analysis of Countermeasures Against the Trojan-Horse Attacks on Quantum Key Distribution Systems in 1260–1650 nm Spectral Range},'' \emph{Optics and Spectroscopy}, vol. 128, pp. 1892 -- 1900, 2020. [Online]. Available: \url{https://api.semanticscholar.org/CorpusID:230659143}
\BIBentrySTDinterwordspacing

\bibitem{Pan2020}
\BIBentryALTinterwordspacing
Y.~Pan, L.~Zhang, and D.~Huang, ``{Practical Security Bounds against Trojan Horse Attacks in Continuous-Variable Quantum Key Distribution},'' \emph{Applied Sciences}, vol.~10, no.~21, 2020. [Online]. Available: \url{https://www.mdpi.com/2076-3417/10/21/7788}
\BIBentrySTDinterwordspacing

\bibitem{Navarrete_2022}
\BIBentryALTinterwordspacing
Álvaro Navarrete and M.~Curty, ``Improved finite-key security analysis of quantum key distribution against trojan-horse attacks,'' \emph{Quantum Science and Technology}, vol.~7, no.~3, p. 035021, jun 2022. [Online]. Available: \url{https://dx.doi.org/10.1088/2058-9565/ac74dc}
\BIBentrySTDinterwordspacing

\bibitem{Sharma2021}
P.~Sharma, A.~Agrawal, V.~Bhatia, S.~Prakash, and A.~K. Mishra, ``{Quantum Key Distribution Secured Optical Networks: A Survey},'' \emph{IEEE Open Journal of the Communications Society}, vol.~2, pp. 2049--2083, 2021.

\bibitem{cao2022evolution}
Y.~Cao, Y.~Zhao, Q.~Wang, J.~Zhang, S.~X. Ng, and L.~Hanzo, ``{The evolution of quantum key distribution networks: On the road to the qinternet},'' \emph{IEEE Communications Surveys \& Tutorials}, vol.~24, no.~2, pp. 839--894, 2022.

\bibitem{Devroye1996}
\BIBentryALTinterwordspacing
L.~Devroye, L.~Gy{\"o}rfi, and G.~Lugosi, \emph{A Probabilistic Theory of Pattern Recognition: The Bayes Error}.\hskip 1em plus 0.5em minus 0.4em\relax New York, NY: Springer New York, 1996. [Online]. Available: \url{https://doi.org/10.1007/978-1-4612-0711-5_2}
\BIBentrySTDinterwordspacing

\bibitem{ANBARASI2017}
\BIBentryALTinterwordspacing
K.~Anbarasi, C.~Hemanth, and R.~Sangeetha, ``A review on channel models in free space optical communication systems,'' \emph{Optics and Laser Technology}, vol.~97, pp. 161--171, 2017. [Online]. Available: \url{https://www.sciencedirect.com/science/article/pii/S0030399216315006}
\BIBentrySTDinterwordspacing

\bibitem{Etxezar2021}
\BIBentryALTinterwordspacing
J.~E. Martinez, P.~Fuentes, P.~M. Crespo, and J.~Garcia-Fr{\'i}as, ``{Time-varying quantum channel models for superconducting qubits},'' \emph{npj Quantum Information}, vol.~7, pp. 1--10, 2021. [Online]. Available: \url{https://api.semanticscholar.org/CorpusID:236096518}
\BIBentrySTDinterwordspacing

\bibitem{huang2018quantum}
A.~Huang, ``Quantum hacking in the age of measurement-device-independent quantum cryptography,'' 2018.

\bibitem{achab2020ranking}
A.~Mastane, ``Ranking and risk-aware reinforcement learning,'' Ph.D. dissertation, Institut polytechnique de Paris, 2020.

\bibitem{Boucheron2010}
\BIBentryALTinterwordspacing
B.~Stéphane, B.~Olivier, and L.~Gábor, ``Theory of classification: a survey of some recent advances,'' \emph{ESAIM: Probability and Statistics}, vol.~9, pp. 323--375, 3 2010. [Online]. Available: \url{http://eudml.org/doc/104340}
\BIBentrySTDinterwordspacing

\bibitem{Yan2017}
H.-C. Yan, J.-H. Zhou, and C.~K. Pang, ``{Gaussian Mixture Model Using Semisupervised Learning for Probabilistic Fault Diagnosis Under New Data Categories},'' \emph{IEEE Transactions on Instrumentation and Measurement}, vol.~66, no.~4, pp. 723--733, 2017.

\bibitem{berger2014kolmogorov}
V.~W. Berger and Y.~Zhou, ``{Kolmogorov--smirnov test: Overview},'' \emph{Wiley statsref: Statistics reference online}, 2014.

\bibitem{moscovich2013}
A.~Moscovich-Eiger, B.~Nadler, and C.~Spiegelman, ``{The Calibrated Kolmogorov-Smirnov Test},'' \emph{arXiv preprint arXiv:1311.3190}, vol.~65, pp. 694--706, 2013.

\bibitem{schaffer1993selecting}
C.~Schaffer, ``Selecting a classification method by cross-validation,'' \emph{Machine learning}, vol.~13, pp. 135--143, 1993.

\end{thebibliography}
\balance
\end{document}